\documentclass[12pt]{article}
\usepackage[utf8]{inputenc}
\usepackage[T1]{fontenc}
\usepackage{lmodern}
\usepackage[margin=1in]{geometry}
\usepackage{setspace}
\usepackage{amsmath,amssymb,amsthm,mathtools}
\usepackage{graphicx}
\usepackage{booktabs}
\usepackage{siunitx}
\usepackage{microtype}
\usepackage{enumitem}
\usepackage[numbers,sort&compress]{natbib}
\usepackage{placeins}
\usepackage[hidelinks]{hyperref}
\usepackage{cleveref}
\usepackage{array}
\usepackage{xurl}
\usepackage{xcolor}
\usepackage{orcidlink}

\theoremstyle{plain}
\newtheorem{theorem}{Theorem}
\newtheorem{lemma}{Lemma}
\newtheorem{proposition}{Proposition}
\newtheorem{corollary}{Corollary}
\theoremstyle{definition}
\newtheorem{definition}{Definition}
\newtheorem{assumption}{Hypothesis}
\theoremstyle{remark}
\newtheorem{remark}{Remark}

\newcommand{\TV}{\mathrm{TV}}
\newcommand{\zL}{\lvert 0_L\rangle}
\newcommand{\oL}{\lvert 1_L\rangle}
\newcommand{\Pzero}{P_{U\mid 0_L}}
\newcommand{\Pone}{P_{U\mid 1_L}}
\newcommand{\Zbar}{\bar{Z}}
\newcommand{\Xbar}{\bar{X}}
\newcommand{\Sgroup}{\mathcal{S}}
\newcommand{\GB}{\mathcal{G}_B}
\newcommand{\Norm}{N(\mathcal{S})}
\newcommand{\code}{\Pi}
\newcommand{\etabar}{\bar{\eta}}
\newcommand{\Rlat}{R_{\mathrm{lat}}}
\newcommand{\Robs}{R_{\mathrm{obs}}}
\DeclareMathOperator{\Tr}{Tr}

\title{\vspace{-2em}\bfseries Execution-transcript privacy for fault-tolerant surface-code memories}
\author{%
Jiachen Shen\,\orcidlink{0000-0002-7233-497X}$^{1}$,\quad
Hui Zhong\,\orcidlink{0009-0005-7952-007X}$^{2,*}$\\[4pt]
{\footnotesize $^{1}$Department of Electrical and Computer Engineering, University of Houston, Houston, TX 77204, USA}\\
{\footnotesize $^{2}$Department of Computer Science and Software Engineering, Miami University, Oxford, OH 45056, USA}\\[2pt]
{\footnotesize $^{*}$Corresponding author: \texttt{zhongh7@miamioh.edu}}}
\date{}

\IfFileExists{build_stamp.tex}{% generated by stamp.py -- do not edit
\newcommand{\buildstampline}{}
}{}

\newcommand{\archivedoi}{\href{https://doi.org/10.5281/zenodo.22102706}{10.5281/zenodo.22102706}}
\providecommand{\buildstampline}{}

\begin{document}
% \maketitle puts the author block in a tabular, which adds 2\tabcolsep of padding around it. The widest
% affiliation is 460.54pt against a 469.76pt text block, so the default 6pt padding is exactly what pushed
% the title block 2.79pt into the margin. Zero it for the title only.
{\setlength{\tabcolsep}{0pt}\maketitle}
\vspace{-3em}

\buildstampline
\begin{abstract}
\noindent
A fault-tolerant quantum computer runs behind a telemetry stream logging syndromes, decoder actions, resets and
timing separately from the answer. Can it reveal the logical input? For a distance-$d$ rotated
surface-code memory on a fixed schedule of $T=\Theta(d)$ rounds, under three stated hypotheses (sector-scalar honest
backbone, transcript locality, Koteck\'y--Preiss smallness), the channel from logical qubit to transcript is $e^{-\Theta(d)}$-close in diamond norm to one that ignores the input. A statement of this kind follows
generically from correctability--privacy duality. Anisotropy does not. Each logical axis pays the distance of its own
coset, so under amplitude damping the computational-basis label is governed by the code's $Z$-distance
$d_Z\ge d_{\min}$ and not by the code distance. Two codes of quantum distance $1$ make the gap concrete. A phase-flip
code's $X$-syndrome transcript is exactly input-independent under unobserved damping, while a repetition code leaks
at first order. A matched converse identifies the records that do expose it, among them a lattice-surgery parity readout.

On a 156-qubit superconducting processor our sufficient certificate misses by $21.5\times$, so the theorem cannot be
invoked there. Measured directly, a $d_Z=1$ memory's record identifies its input with total variation $\ge0.927$
under randomised, label-balanced acquisition. Holding the code fixed and varying the damping exposure reproduces the
parameter-free law, with exponent $0.85\pm0.03$ against a predicted $0.86$. Randomized encoding returns the statistic
to the floor at no two-qubit-gate cost. Fault tolerance does not grant transcript privacy. It relocates it, and only
to the logical state, not to the circuit's identity.
\end{abstract}

\section*{Introduction}
Fault-tolerant quantum computers run behind a thick layer of classical telemetry. Stabilizer error correction and its threshold theorems make large computations possible~\cite{shor1995,gottesman1997,preskill1998,kitaev2003,aharonov2008,
aliferis2006,raussendorf2007,wang2003,terhal2015,campbell2017,gottesman2010,bravyi2005,nayak2008}, and small surface-code memories now run in the laboratory~\cite{krinner2022,google2023,acharya2024}. Every logical operation incurs the easy-to-overlook cost of many error-correction rounds, which emit a continuous stream of syndromes and decoder choices. The continuous stream also records each ancilla reset and carries both timing information and leakage flags~\cite{dennis2002,fowler2012,gidney2021stim}. A distance-$d$ memory run for $\Theta(d)$ rounds already produces
$\Theta(d^3)$ such records, and that stream is logged and shared with a control plane because real-time decoding
depends on it. It is a large, permanent interface to a running computation, separate from its answer.

The concrete question is this. Anyone who later reads the provider's control logs also sees the entire syndrome and decoder stream. A cloud service runs a customer's circuit and returns only the final logical answer, but the provider sees the same stream while the computation runs. Can the provider tell whether a memory held a logical $0$ or a logical $1$? Intuition pulls both
ways. Syndromes are by design uncorrelated with the logical state, which suggests safety, but the transcript is enormous and has repeated structure. Under these conditions, classical side channels have defeated systems that looked safe in isolation~\cite{kocher1996}. Stabilizer folklore settles the noiseless case~\cite{knill1997,nielsen2010}, but it neither provides a theorem for the multi-round observed transcript nor addresses coherent or non-Markovian records, leaving the safe records unidentified. A lattice-surgery logical measurement is
built from the same syndrome machinery, yet is designed to reveal a logical value.

Part of the answer is already available. Correctability--privacy duality~\cite{kretschmann2008,beny2010} implies that a below-threshold memory's transcript is $e^{-\Theta(d)}$-private. The transcript is a register on the complementary side of the Stinespring dilation, so accurate recovery of the logical state from the memory output bounds what the complementary register can distinguish. The Relation section builds that dilation and its cut explicitly. We therefore present this corollary as an existing consequence of the duality. What it does not supply is the \emph{parameter}. Matched simulation reproduces the state-dependent signal of the hardware experiment reported here from amplitude damping alone. Under amplitude damping, the basis-label leak order is the code's $Z$-distance $d_Z$, not the code distance $d_{\min}=\min(d_X,d_Y,d_Z)$, and the two can be arbitrarily far apart. A phase-flip code has quantum distance $1$, so no correctability statement exists and the duality is vacuous. Its $X$-syndrome transcript is \emph{exactly} input-independent under unobserved damping, whereas a repetition code with the same quantum distance leaks at first order, $nT\gamma+O(\gamma^2)$. Reading the code distinguishes the two cases because correctability alone does not.

This paper makes five contributions. \emph{One}, it frames the non-output execution transcript as a privacy channel
with an explicit quantum threat model. \emph{Two}, it identifies the governing parameter of the basis-label leak and determines separately when that order is \emph{attained}. No term appears below order $\gamma^{\,d_Z}$, with $d_Z$ computed by enumerating the logical coset in the binary-symplectic representation of the Pauli group. The coefficient is $nT$ for the repetition family. For a measured algebra of $X$ type, meaning that every operator the instrument actually measures is a product of $X$ operators, the leak vanishes identically and the order is never attained. This exhibits codes with a vacuous duality bound but perfect basis-label privacy in the syndrome record, so the exposed record set determines which case applies. Once the environment's jump record is also exposed, the basis label leaks at exactly $\gamma^{\,d_Z}$ already at one round, and $d_Z$ governs either case. \emph{Three}, it proves a no-leakage theorem under two transcript-locality hypotheses and a polymer-smallness condition. In diamond norm, the induced channel from the logical qubit to the transcript is exponentially close to input-independent across arbitrary logical inputs, without restriction to a basis pair. A witness that is \emph{charged}, meaning it carries a logical operator rather than a stabilizer, and \emph{calibrated}, meaning its recorded value tracks that operator, leaks the input at a rate set by the distance. A lattice-surgery parity readout provides one example. \emph{Four}, it proves this converse and supplies a diagnostic that screens a given record set. \emph{Five}, on a 156-qubit superconducting processor it measures the predicted $d_Z$ mechanism in the repetition family, and it separately evaluates our sufficient certificate for the smallness condition, which misses by $21.5\times$. Under randomised, label-balanced acquisition, a $d_Z=1$ memory's honest transcript reaches a total-variation lower bound of $0.927$. Physical basis-state randomized encoding equalises the two transcript laws by construction at no cost in additional two-qubit gates, thereby returning the measured statistic to its finite-sample floor.

The operational principle is that fault tolerance does not grant transcript privacy. Instead, fault tolerance relocates it. Under the transcript hypotheses stated below, a code with large $d_{\min}$ protects an arbitrary logical state, and one with large $d_Z$ protects the computational-basis label under amplitude damping. Records on a short, identifiable list behave as logical outputs and must be guarded as such, but this relocation remains incomplete for the codes running on present hardware. The rest of the paper makes each statement precise. The Supplementary Information derives the locality hypothesis from the device model and includes both the full proofs and extended numerics.

\section*{Results}

\begin{figure}[t]
\centering
\includegraphics[width=0.82\linewidth]{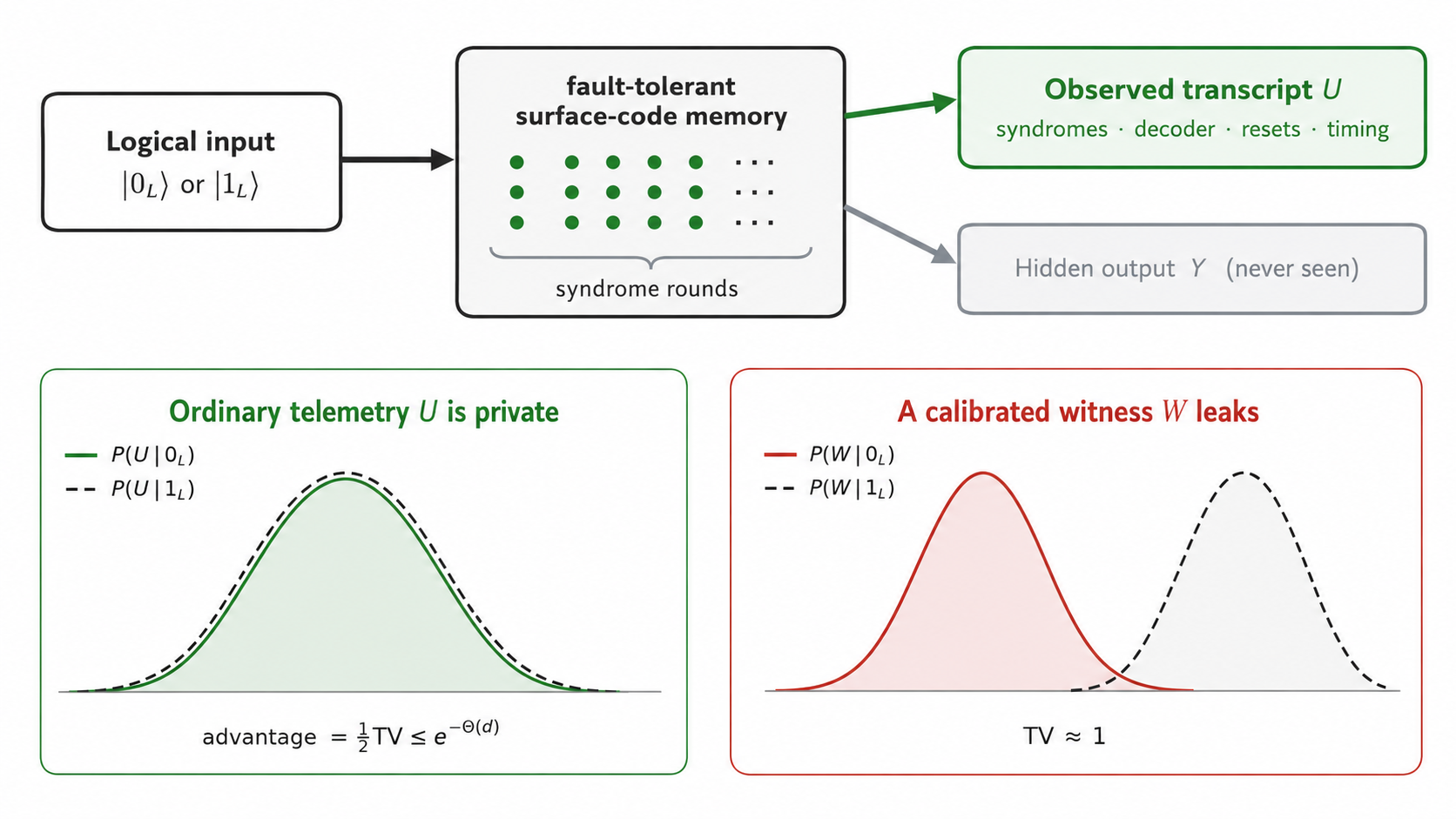}
\caption{The transcript privacy channel and the two regimes. A logical input enters a surface-code memory that emits the
non-output transcript $U$ while the output $Y$ stays hidden. The two input-conditioned laws of ordinary $U$ nearly
coincide (private), while a calibrated witness $W$ separates them (leaks).}
\label{fig:concept}
\end{figure}

\paragraph*{Guarantee specification.} The forward theorem concerns the non-output transcript of a distance-$d$ rotated surface-code memory on a fixed, input-independent schedule under the two transcript-locality hypotheses and the smallness condition stated below. The converse supplies a sufficient leaking condition. The protected object consists of the execution transcript and excludes the final output, maliciously chosen schedules, direct physical probes of data qubits, and records calibrated to logical observables. Other codes, schedules and records enter through the same theorem hypotheses. The diagnostic supplies the screening evidence specified below, and certification also requires the theorem's assumptions.

\paragraph*{Two privacy notions.} Two quantities are kept apart throughout, and both are named here because both are used before the Relation section compares them with what correctability supplies. \emph{Full logical-channel privacy} is the diamond-norm distance between the transcript channel and an input-independent one. It covers arbitrary logical inputs, including superpositions, mixtures and reference-entangled states, and the ordinary code distance $d_{\min}$ governs it. \emph{Basis-label privacy} is the total variation between the transcripts of the two $\Zbar$ eigenstates under the damping model specified below. It protects one logical axis instead of the whole channel, so it is the weaker of the two, and $d_Z$ governs the first order at which it can fail.

\subsection*{The execution transcript is a privacy channel}
The execution transcript defines a classical privacy channel from the logical input to the observed telemetry,
separate from the hidden logical output. While a fault-tolerant computer runs, it emits a stream of classical records, and these include syndrome outcomes,
decoder corrections, ancilla resets, timing records, and leakage flags. We call this
record the execution transcript $U$ and ask whether it reveals the logical input.

The setting is drawn in \cref{fig:concept}. The memory starts in one of the two logical basis states $\zL$ or $\oL$,
and runs a distance-$d$ rotated surface code for $T=\Theta(d)$ rounds. It produces two things. One is the logical
output $Y$, which goes to the user and which the adversary never sees, while the other is the non-output transcript $U$,
the telemetry stream, which the adversary sees in full. The threat model is deliberately generous to the adversary because it is given the entire transcript, unlimited classical computation, and full knowledge of the circuit, the noise model, and the decoder. It is denied only the hidden output and any direct access to the data qubits, and this is the position of an
honest-but-curious cloud provider, or of anyone who later obtains the provider's control logs.

We measure leakage by the total variation distance between the two transcript laws $\Pzero$ and $\Pone$. By Neyman--Pearson this equals the gap in optimal distinguishing power, so it bounds every adversary at once instead of only a chosen statistic, and it assumes nothing about the adversary's prior or computation. It is the right metric here because the adversary makes a single binary decision. For example, the max-divergence behind pure differential privacy weights rare events heavily and would raise an alarm on transcripts of exponentially small probability that never affect that decision. The Methods record the resulting approximate-privacy corollaries.

\subsection*{Ideal syndrome transcripts are independent of the logical input}
In the idealized stabilizer model the transcript distribution is the same for both logical inputs. The reason is structural because the two logical states differ by the valid logical operator $\Xbar$, which commutes with every stabilizer the device measures, so applying it changes no syndrome and is invisible to the detectors. The transcript
therefore cannot depend on whether $\Xbar$ was applied, so $\Pzero=\Pone$ exactly.

This elementary base case gives exact privacy in the noiseless limit and follows from the same algebra that makes the code correct errors. It does not say every record is safe because a direct readout of a data qubit is not a stabilizer measurement and need not commute with $\Xbar$. The content of the theorem is to carry this into the noisy
regime and mark where it breaks.

\subsection*{The mechanism: logical charge costs distance}
Any part of the transcript that can distinguish the two logical inputs must carry logical charge, meaning a nonzero component in a nontrivial logical Pauli coset, and in a topological
code that charge costs distance (\cref{fig:mechanism}).

\begin{figure}[t]
\centering
\includegraphics[width=0.9\linewidth]{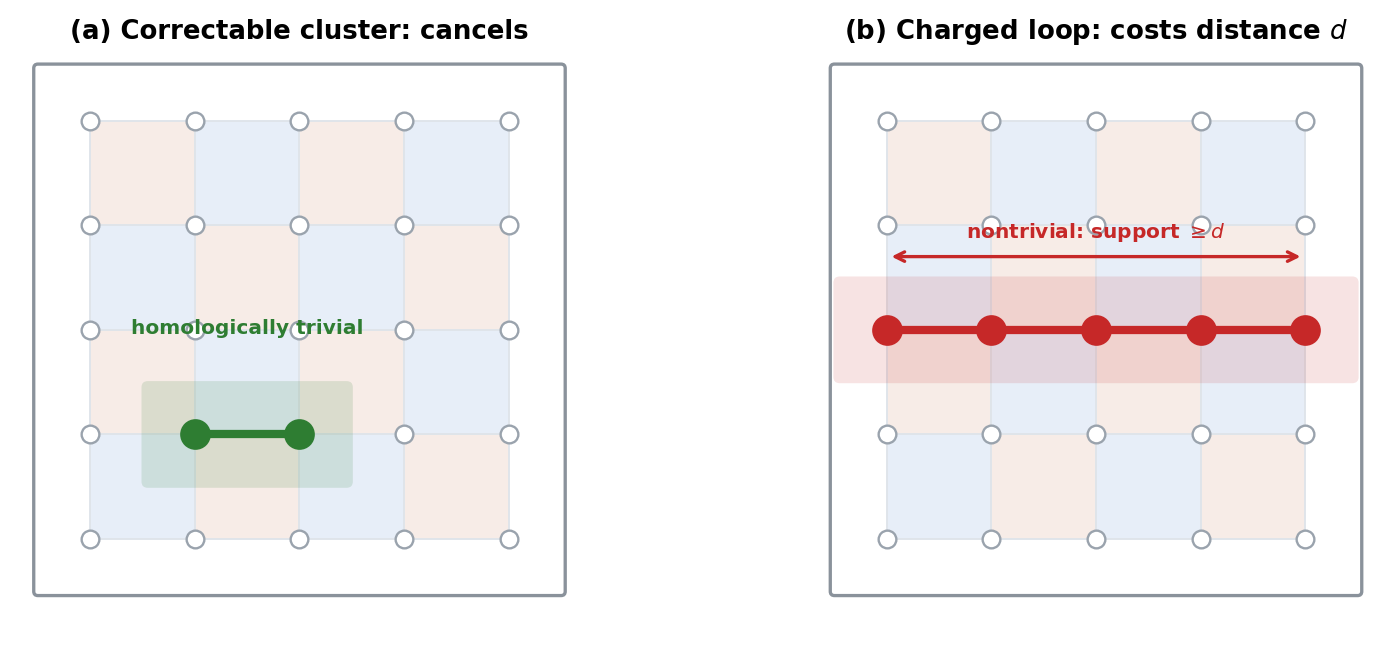}
\caption{The mechanism. (a) By Knill--Laflamme, a bra--ket fault pair $K_b^\dagger K_a$ with homologically \emph{trivial} support multiplies the code space by a scalar within its syndrome sector and cancels in the two transcript laws. The condition is trivial homology, not weight below $d$ (\Cref{rem:supports}). (b) An input-sensitive contribution
must therefore carry a homologically nontrivial cycle, whose support is $\geq d$, so logical information enters only
through a charged region of size $\geq d/\ell_\ast$. The schematic is spatial only, and the expansion runs over spacetime detector regions.}
\label{fig:mechanism}
\end{figure}

Precisely: expand the noisy transcript into fault paths and group contributions by the spacetime detector region they
produce. By Knill--Laflamme, a homologically trivial bra--ket pair is proportional to the identity on the logical
subsystem within its syndrome sector, so it contributes equally to $\Pzero$ and $\Pone$ and cancels. A logical string spanning the patch or wrapping the geometry is a homologically nontrivial cycle, and a contribution survives in the difference only if it carries such a cycle whose support in the surface code is at least $d$. Here $\ell_\ast$ is the cell-support constant, the largest number of data qubits that any one detecting cell covers. A charged operator of weight $d$ therefore cannot be carried by fewer than $d/\ell_\ast$ cells, so a surviving input-sensitive contribution occupies a region at least that large. Under the polymer smallness condition the summed activity of all regions that large is exponentially small in $d$. The constants in that sum are independent of $d$ and $T$, which is the uniformity the theorem needs.

The two facts doing the work here are Knill--Laflamme and the homology of the code, and neither fact concerns decoders, timing or hardware, so the classical post-processing a real device applies to its syndromes cannot change the conclusion. This post-processing is a fixed channel, and a fixed channel cannot increase distinguishability.

\subsection*{The order of the basis-label leak, and when it is attained}
\label{sec:exactscaling}
The mechanism above bounds the \emph{order} of the leak by a quantity computed instead of sampled, which matters because a classifier applied to simulated records can only ever report a lower bound on its own resolution.
Throughout this subsection the quantity is basis-label privacy, in the sense named above, and not full logical-channel privacy.

Amplitude damping acts in the physical computational basis in the following sense. Its Kraus operators are $K_0=\lvert0\rangle\!\langle0\rvert+\sqrt{1-\gamma}\,\lvert1\rangle\!\langle1\rvert$ and $K_1=\sqrt{\gamma}\,\lvert0\rangle\!\langle1\rvert$, and each carries a computational basis state to a computational basis state up to a scalar. A product of them therefore maps the basis to itself, so every $Z$-type stabilizer outcome probability is a sum of expectations $\langle Z_S\rangle$ over subsets $S$ of the data qubits. The transcript therefore depends on the logical bit only through those moments, which are nonzero only when $Z_S\in\langle\Sgroup,\Zbar\rangle$ and \emph{differ} between the inputs only on the coset $\Zbar\cdot\Sgroup$. The lowest weight at which the two
inputs can differ is therefore the minimum weight of a logical-$Z$ representative, the code's $Z$-distance $d_Z$.

\begin{remark}[what $U$ contains in this subsection]\label{rem:scope}
The general transcript of the Methods is $U=(S,A,R,\tau,J)$, and $J$ there can include genuine jump or herald records. \emph{Throughout this subsection $U$ is the projective stabilizer-syndrome record together with its
deterministic classical post-processing, and nothing else}. Neither an environment jump trajectory nor a data-local herald is exposed, and no analog amplitude record is exposed. The distinction is not cosmetic. A jump record is $Z$-diagonal and is therefore not a function of an $X$-type measured algebra, so \Cref{prop:xonly} below would be false for it, which \Cref{prop:jump} shows quantitatively.
\end{remark}

\begin{proposition}[no leak below order $d_Z$]\label{prop:dzlaw}
Fix a CSS code, a number of rounds $T$, and the syndrome-only transcript of \Cref{rem:scope}, and let each data
qubit undergo amplitude damping with per-round probability $\gamma$. Let $\rho_{0_L},\rho_{1_L}$ be the two $\Zbar$-eigenstates.
Then as $\gamma\to0$ at fixed $(d,T)$,
\begin{equation}\label{eq:dzlaw}
\TV\bigl(P^{(T)}_{0_L,\gamma},\,P^{(T)}_{1_L,\gamma}\bigr)=O_{d,T}\!\bigl(\gamma^{\,d_Z}\bigr),
\end{equation}
where $d_Z$ is the minimum weight of a representative of the coset $\Zbar\cdot\Sgroup$. No term of order below $\gamma^{\,d_Z}$ is algebraically permitted because every moment $\langle Z_S\rangle$ with $|S|<d_Z$ agrees on the two inputs. Writing the expansion as $c_{d,T}\gamma^{\,d_Z}+O_{d,T}(\gamma^{\,d_Z+1})$, the statement is that $d_Z$ is the first \emph{permitted} order. It does not assert $c_{d,T}>0$.
\end{proposition}

We are careful here because the two halves of ``the leading order is $d_Z$'' have very different status, since the algebra forces the absence of terms below $d_Z$ but does not force any term to appear \emph{at} $d_Z$. A minimum-weight charged event can produce no detector response, or several can cancel, in which case the true leading order is higher, and both possibilities are realised, so neither can be waved away.

\begin{proposition}[the order is attained for the repetition family]\label{prop:attained}
For the bit-flip repetition code on $n$ qubits with $\Zbar=Z_1$, under the damping-only transcript model,
\begin{equation}\label{eq:repclosed}
\TV\bigl(P^{(T)}_{0_L},P^{(T)}_{1_L}\bigr)
=1-(1-\gamma)^{nT}-\gamma^{\,n}\!\!\sum_{k=0}^{T-1}\!(1-\gamma)^{nk}
=nT\,\gamma+O(\gamma^{2}),
\end{equation}
so $c_{n,T}=nT>0$ and the order $d_Z=1$ is exact.
\end{proposition}
\begin{proof}
$\zL=\lvert0\cdots0\rangle$ is a fixed point of amplitude damping, so its record is deterministically all-zero and
$\TV=1-\Pr[\text{all-zero}\mid\oL]$. From $\oL=\lvert1\cdots1\rangle$ every check $Z_iZ_{i+1}$ reads $+1$ exactly when all qubits agree, i.e.\ when the decayed set is empty or everything. Since decays only accumulate, an all-zero record means qubits survive intact for $k$ rounds and then decay together in one round, so summing over $k$ gives \cref{eq:repclosed}, while the $\gamma^{\,n}$ term is the all-decay branch and is silent again.
\end{proof}

\begin{proposition}[exact blindness of an $X$-only \emph{syndrome} transcript under unobserved damping]
\label{prop:xonly}
Let $U$ be the syndrome transcript of \Cref{rem:scope}, and suppose every measured stabilizer is $X$-type, which is the case for the phase-flip code because its $Z$-stabilizer group is trivial. Then under amplitude damping
\begin{equation}\label{eq:xonlyzero}
\TV\bigl(P^{(T)}_{0_L},P^{(T)}_{1_L}\bigr)=0\qquad\text{exactly, for every }\gamma\text{ and every }T,
\end{equation}
so $c_{d,T}=0$ and $\gamma^{\,d_Z}$ is an upper bound that is not attained. The hypothesis that the damping jumps are
\emph{unobserved} is essential and is exactly what \Cref{rem:scope} fixes.
\end{proposition}
\begin{proof}
The adjoint of amplitude damping maps an $X$-type Pauli to a scalar multiple of itself, $\mathcal{D}^\dagger_\gamma(X_A)=(1-\gamma)^{|A|/2}X_A$, without mixing in anything new. Only $Z$ acquires an identity component. Every syndrome projector of an $X$-only code is a linear combination of $X$-type stabilizer elements, and
each such element has expectation $\pm1$ on \emph{both} codewords, so every outcome probability
$\Tr(\Pi_s\mathcal{D}_\gamma(\rho_b))$ is independent of $b$. A stabilizer element multiplies states in the projected sector by its eigenvalue, so the property is preserved by the measurement update and the induction closes over rounds.
\end{proof}

That code has $d_X=1$, hence no correctability guarantee whatever, yet its basis-label \emph{syndrome} transcript is perfectly private instead of merely private to order $\gamma^{\,n}$. After one damping layer at $\gamma=0.1$, the two damped codewords are $1.70$ apart in trace distance at $n=3$, near the maximum of $2$. The states are not close, so a direct measurement separates them almost perfectly while the syndrome record carries nothing at all.

\begin{proposition}[exposing the jump record restores the leak, at exactly order $d_Z$]\label{prop:jump}
Let $U'=(U,J)$ adjoin to the syndrome record the environment's which-qubit-decayed pattern in each round. For the
phase-flip code, already at $T=1$,
\begin{equation}\label{eq:jumpleak}
\TV\bigl(P'_{0_L},P'_{1_L}\bigr)=\gamma^{\,n}\bigl(1+O(\gamma)\bigr),\qquad n=d_Z,
\end{equation}
so exact blindness is lost, the order is $d_Z$, and the coefficient is $1$.
\end{proposition}

Exact computation gives $\TV=1.0000\times10^{-6}$ at $n=3,\gamma=10^{-2}$ and $1.0000\times10^{-8}$ at
$n=4,\gamma=10^{-2}$, with measured log-log slopes $3.0000$ and $4.0000$ against $d_Z=3,4$. The phase-flip family therefore exhibits all three regimes at once, and the only thing that changes between them is which records are exposed. The two damped \emph{states} are $1.70$ apart, while the \emph{syndrome} transcript is exactly blind and the syndrome \emph{plus jump} transcript leaks at exactly $\gamma^{\,d_Z}$. $d_Z$ is the governing parameter throughout. This is the relocation thesis of this paper, demonstrated inside a single code family instead of argued for.

\paragraph*{The order, measured exactly.} The exact computation uses full density-matrix evolution with projective stabilizer measurement and enumerates every outcome branch without sampling. Because attainment is now a question instead of an assumption, we compute $\TV(\gamma)$ exactly and read off the local slope $\mathrm{d}\log\TV/\mathrm{d}\log\gamma$ as $\gamma\to0$ (\texttt{sim/hw/damping\_order.py}). The total variation grows with the number of recorded rounds, so the table carries the round count $T$ used for each row.

\begin{center}\small
\begin{tabular}{lccccc}
\toprule
code & $d_Z$ & $M_d$ & $T$ & $\TV$ at $\gamma=10^{-4}$ & measured slope\\
\midrule
repetition $n=3$ & 1 & 3 & $2$ & $5.9985\times10^{-4}$ & $0.9995$\\
repetition $n=5$ & 1 & 5 & $2$ & $9.9955\times10^{-4}$ & $0.9992$\\
$[[4,1,2]]$ & 2 & 4 & $2$ & $7.9974\times10^{-8}$ & $1.9994$\\
rotated surface $d=3$ & 3 & 8 & $1$ & $6.9981\times10^{-12}$ & $2.9995$\\
\addlinespace
phase-flip $n=3,4,5$ & $n$ & 1 & any & $<10^{-16}$ & --- ($\TV\equiv0$)\\
\bottomrule
\end{tabular}
\end{center}

The slope converges to $d_Z$ in every case where the coefficient is nonzero, confirming \Cref{prop:attained} and extending the result to two codes for which we have no closed form. It is absent exactly where \Cref{prop:xonly} says it must be, and $M_d$, the multiplicity of minimum-weight representatives, is reported as a computed structural quantity and must \emph{not} be read as $c_{d,T}$. For the repetition family $c_{n,T}=nT$, not $M_d=n$, whereas for the phase-flip family $M_d=1$ while $c_{d,T}=0$.

Three consequences follow.

First, for the rotated surface code $d_Z=d$, so the first permitted order rises with the code distance, and the structural proxy $M_d\gamma^{d_Z}$ also decreases by $1538\times$ between $d=3$ and $d=5$ at $\gamma=10^{-2}$. This is a statement about that proxy, not a measured drop in $\TV$, because $c_{d,T}$ is undetermined at $d=5$. The multiplicity works against the
decay, growing from $M_3=8$ to $M_5=52$, an effective finite-distance growth factor
$\lambda_{\rm eff}=\sqrt{52/8}\approx2.55$ per unit distance, which at $\gamma=10^{-2}$ is overwhelmed by the
$\gamma^{2}$ gained.

We stop short of an asymptotic claim, and the reason is a limit order that is easy to cross by accident.
\Cref{prop:dzlaw} is a statement about $\gamma\to0$ at \emph{fixed} $(d,T)$, and it therefore does not, on its own, say anything about $d\to\infty$ at fixed $\gamma$. Passing between the two would require control of $c_{d,T}$ and of the remainder as $d$ grows, and we have neither. The exponential-in-$d$ statements in this paper come from
\Cref{thm:main}, whose constants are uniform by construction, and not from \cref{eq:dzlaw}. Separately, two distances cannot determine the growth rate of $M_d$, and we do not extrapolate $\lambda_{\rm eff}$ into a statement of the form $\TV=\Theta((\lambda\gamma)^d)$ or a threshold in $\gamma$. Accordingly, the asymptotic growth of $M_d$ remains open.

Second, $d_Z$ is not the code distance, and the difference is the whole point. Since $d_{\min}=\min(d_X,d_Y,d_Z)$ we have $d_Z\ge d_{\min}$, so a code can be a poor quantum memory and still hide its \emph{computational-basis label} very well. The phase-flip family is the extreme case. It has $d_X=1$, hence no correctability guarantee and no full logical-channel privacy whatever, while \Cref{prop:xonly} gives a basis-label transcript that is exactly input-independent instead of only $O(\gamma^{\,n})$. The Relation section works this family through, including the codewords, which are cat states instead of the product states one first guesses.

Third, the repetition code has $d_Z=1$ at every $n$. Its basis-label leak is \emph{first order} in $\gamma$ and does not improve with length, and by \Cref{prop:attained} that order is attained with coefficient $nT$. This is not a
defect of the experiment that uses it but the reason for using it, since the obvious objection to a repetition-code demonstration states that $\lvert1\cdots1\rangle$ relaxes and consequently leaks. It is the degenerate end of the same criterion and the worst case for damping-induced charged leakage, and it converts that objection into the paper's own prediction. The surface code sits at the
other end. The comparison is evaluated on hardware in the Hardware measurement section.

The computed quantities come from exact computation instead of sampling. $d_Z$ and $M_d$ are verified two independent ways at every size the algebra can reach, by GF(2) enumeration and by dense statevector simulation, agreeing on the order \emph{and} the multiplicity. Each ratio below is the GF(2) count over the statevector count, and the four codes give $3/3$, $5/5$, $4/4$ and $8/8$. The slopes above come from a third, independent construction. Enumeration beyond the stated rank cap is refused instead of approximated, because an under-search would return a $d_Z$ that is too \emph{large} and would flatter the privacy claim.

\subsection*{Ordinary noisy transcripts are exponentially indistinguishable, axis by axis}
Under H1$^\ast$, PWTS and the Koteck\'y--Preiss smallness condition, the noisy transcript is private up to a distinguishing advantage that is exponentially small while the exponent is not one number. Each logical axis pays the distance of its own coset. That anisotropy is the new content of this section. The familiar single-exponent statement follows from it by taking a minimum, and in the correctable regime it is a corollary of the duality (\Cref{prop:corollary}), not of this paper.

The proof turns the mechanism above into a quantitative bound by a cluster expansion over spacetime detector regions
(Methods and Supplementary Note~11). Making the suppression uniform in the spacetime volume is the technical heart of
the argument, and it is where the two hypotheses enter. Syndrome extraction with reset ancillas and decoder post-processing form the error-free part of the device. Under \emph{Backbone blindness} (H1$^\ast$), this part produces the same records for every logical input because the measured operators multiply the code space by scalars, so error-free extraction contributes no logical signal. \emph{Local transcript summability} (PWTS) says every connected region of the record has
an activity that decays geometrically with the region's size, uniformly in $d$ and $T$, so the rare charged regions do
not proliferate. The closed bound needs one more quantitative condition, that this decay rate be small enough (the
Koteck\'y--Preiss smallness $ea<1$ of the Methods), so that the expansion converges. The records that pass are reset-ancilla syndrome detectors, decoder actions, and input-independent timing and reset metadata, while the records that fail include a lattice-surgery logical measurement and a persistent sensor calibrated to a data qubit. They also include a direct quantum-nondemolition readout whose calibrated reports combine along a logical string and the other records that the device map below places on the leaking side.

\begin{theorem}[anisotropic transcript-privacy threshold]\label{thm:main}
Recall the cell-support constant $\ell_\ast$ from the mechanism above. Let the memory be a \emph{single-logical-qubit, geometrically local} stabilizer family on a fixed input-independent schedule with $T=\Theta(d)$ rounds, whose detector graph has bounded degree and for which this constant is independent of $d$ and $T$.
The rotated surface code is such a family. Assume backbone blindness (H1$^\ast$) and local transcript summability
(PWTS), together with the below-threshold smallness condition on the activity $a=e\nu K_{\rm rod}\etabar$ of a polymer, which here is a connected cluster of detector events
(the Koteck\'y--Preiss condition $ea<1$ of the Methods). Here $\etabar<1$ is the geometric decay rate of PWTS and $K_{\rm rod}$ its per-cell prefactor, while $\nu$ is the connective constant of the region graph, the number of connected regions of a given size anchored at one cell. All three are independent of $d$ and $T$. Write the code-space block of each transcript effect in the
logical Pauli basis,
$M_u=q_uI_L+x_u\Xbar+y_u\bar Y+z_u\Zbar$. Then \emph{each logical axis is suppressed at the distance of its own
coset}. The bounds are
\begin{equation}\label{eq:thm-aniso}
\sum_u|x_u|\le\kappa_{d_X},\qquad
\sum_u|y_u|\le\kappa_{d_Y},\qquad
\sum_u|z_u|\le\kappa_{d_Z},
\qquad
\kappa_m:=C_0\,\mathrm{poly}(d,T)\,\theta^{m/\ell_\ast},
\end{equation}
where $d_P$ is the minimum weight of a representative of the coset $\bar P\cdot\Sgroup$, and the constant $C_0$ and threshold $\theta<1$ are independent of $d$ and $T$.
\end{theorem}

The quantitative bound uses the locality hypotheses in the step $|A|\le\ell_\ast|R|$. It is stated for families with uniformly bounded check support and detector-graph degree. \Cref{lem:cosetweight} applies to every stabilizer code, while the polymer bound also assumes the stated geometry.

Two corollaries follow, corresponding to the two privacy notions named above, and because $\kappa_m$ decreases in $m$, bounding all three axes by the smallest of the three distances gives the familiar single-exponent statement.

\begin{corollary}[full logical-channel privacy]\label{cor:fullchannel}
Let $\mathcal{M}$ be the induced classical channel from the logical qubit to the transcript and $\mathcal{C}$ the
input-independent channel emitting a fixed law. Logical inputs include basis states, superpositions, mixtures, or halves of an entangled pair. Then
\begin{equation}\label{eq:cor-diamond}
\|\mathcal{M}-\mathcal{C}\|_\diamond\ \le\ \kappa_{d_X}+\kappa_{d_Y}+\kappa_{d_Z}\ \le\ 3\,\kappa_{d_{\min}}
=e^{-\Theta(d_{\min})},\qquad d_{\min}=\min(d_X,d_Y,d_Z),
\end{equation}
so for \emph{any} two logical inputs $\rho,\sigma$, $\TV(P_\rho,P_\sigma)\le3\kappa_{d_{\min}}$, and every adversary distinguishing them from the transcript succeeds with advantage over guessing at most $\tfrac32\kappa_{d_{\min}}$. For the rotated surface code
$d_{\min}=d$ and this is $e^{-\Theta(d)}$.
\end{corollary}

\begin{corollary}[basis-label privacy]\label{cor:basislabel}
For the two $\Zbar$-eigenstates, $\Tr\bigl(M_u(\rho_{0_L}-\rho_{1_L})\bigr)=2z_u$ identically, so the $x_u$ and $y_u$
terms are absent and
\begin{equation}\label{eq:cor-basis}
\TV(\Pzero,\Pone)=\sum_u|z_u|\ \le\ \kappa_{d_Z}.
\end{equation}
A code with $d_Z\gg d_{\min}$ can therefore be certified to hide its computational-basis label exponentially better than it hides an
arbitrary logical input.
\end{corollary}

\Cref{thm:main} is strictly stronger than \Cref{cor:fullchannel}, and \Cref{cor:basislabel} is the half a single
recovery error cannot reach. The extension from the basis pair to arbitrary inputs is not a triangle inequality over basis states, because that inequality would fail. The measurement $\{(I\pm\Xbar)/2\}$ has identical laws on $\zL$ and $\oL$ yet separates $\lvert\pm_L\rangle$ perfectly. It follows instead from controlling the off-diagonal code-space elements
directly, which is what \cref{eq:thm-aniso} does. \Cref{thm:main} is proved in the Methods, with the constants fixed
in Supplementary Note~11.3 and 11.9--11.10.

The privacy guarantee is the operational one. Because total variation controls every binary test, the bound is an
$f$-differential-privacy statement (Methods) that, for this binary transcript experiment, implies the
approximate-differential-privacy corollaries, including $(0,e^{-\Theta(d)})$-differential privacy. We state it as a
distinguishing-advantage bound because that is what an adversary actually faces. We do not claim pure
$\varepsilon$-differential privacy. Rare transcripts can have an unbounded likelihood ratio, even though their total
probability is exponentially small, so a small additive term is unavoidable and honest. In the running example, for a family meeting the theorem's hypotheses, this says the provider's confidence in reading the stored logical value off the honest syndrome log decays exponentially as the customer buys a higher-distance encoding.

\subsection*{A charged calibrated record leaks}
The same analysis gives a sharp sufficient condition for the guarantee to fail. If the transcript exposes a record
calibrated to a logical operator whose value it recovers with error $p_W$, the logical input leaks at total variation at least $1-2p_W$, and for a witness whose $p_W$ decays with distance the leak again scales with the code
distance, now toward full disclosure.

The cleanest example is a lattice-surgery parity measurement in which merging two patches to read the joint logical parity $\Zbar_1\Zbar_2$ places a logical operator into the measured algebra, which violates H1$^\ast$. The merged-parity
outcome is a near-perfect report of the logical input, and a transcript that contains it is almost perfectly
distinguishing. The converse lemma (Methods) makes this precise. If calibrated local reports multiply to a logical
operator, then the transcript estimates that logical value, and the total variation approaches one as the readout
errors shrink. The converse classifies charged calibrated witnesses. Records that merely correlate with local noise or measure the wrong operator fall outside this class and may remain private. The leak
is established by a witness that is both charged, meaning homologically nontrivial, and calibrated, meaning that its value
tracks a known logical operator.

The leaking side is a designed primitive, not a contrived pathology, because a logical measurement is supposed to reveal a logical value. What the converse adds is that this informativeness \emph{sharpens} with the code distance, the same parameter that hides the honest transcript. The merged parity is protected by a seam of length $d$, so it fails only when an error chain crosses that seam. Theorem and converse run the same engine in opposite directions, and the design instruction is single. Distance suppresses accidental logical strings in the transcript but makes an intentional logical readout more reliable, so keep logical operators out of the exposed algebra.

The records of an honest memory then fall into three groups, which is the device map the rest of the paper uses.
The physical noise of an honest run includes circuit-level depolarizing noise, small coherent over-rotations, finite-range crosstalk, and reduction-handled leakage. Under the two hypotheses, \emph{Certified private} records include reset-ancilla syndrome detectors, which have bounded footprint and are sector-scalar under H1$^\ast$. They also include the decoder's own actions $A=f(S)$ and fixed timing and reset flags, which are input-independent classical kernels. This physical noise is also included because each contribution contains only bounded-footprint fault atoms whose charged sector still costs $d/\ell_\ast$. Full certification also needs
the stated factorization and smallness assumptions, while only records satisfying both qualifiers of the converse are \emph{Leaking}, including a lattice-surgery measurement whose measured algebra contains $\Zbar_1\Zbar_2$. Others are a per-data-qubit QND or dispersive readout, a persistent data-calibrated sensor, or a mobile leakage excitation carrying a record. Each of these records leaks once its local reports combine along a logical string over the $T=\Theta(d)$ rounds, which is the sufficient condition the converse supplies. A local readout that does not so combine, measures the wrong operator, or is one-time-padded is not placed in this group by that condition. \emph{Nuisance} records, including input-independent $1/f$ or common-mode drift, break locality without carrying a logical signal and degrade the logical error rate without revealing the logical bit, so conditioning handles them.

\subsection*{Simulation consistency checks}
The simulations below are consistency checks, and one of them must be labelled precisely, because taken at face value
it would be much stronger evidence than it is. In the stabilizer setting the second logical input is prepared by a verified logical-$\Xbar$ injection, and $\Xbar$ commutes with every measured stabilizer, so the two detector records are drawn from \emph{the same distribution by construction}. Thus $\TV=0$ exactly, as a matter of algebra, before any noise is specified. No linear or other classifier can return any other result, so an adversary returning chance accuracy on that pair has confirmed an identity of the simulator, not measured a privacy property of a device. We therefore read the honest-arm rows of \cref{tab:numerics} as a pipeline check because the analysis chain does not manufacture a separation where none exists, which is what a negative control should establish and no more. The charged-witness
rows are the informative ones, showing that the same chain detects a leak that is genuinely present. Evidence about a real device cannot come from this construction and is reported in the Hardware measurement section, while evidence for the \emph{scaling} of the leak comes from the enumeration above, which is exact instead of sampled.

All numerics use the rotated surface code with circuit-level depolarizing noise, fixed random seeds, and bootstrap
confidence intervals, and every figure is regenerated by a single script from a committed result file. The numerical
values are collected in \cref{tab:numerics}.

\begin{table}[ht]
\centering
\small
\begin{tabular}{@{}p{0.40\linewidth} p{0.21\linewidth} p{0.33\linewidth}@{}}
\toprule
Quantity & Regime & Result (with $95\%$ CI) \\
\midrule
Adversary AUC on honest $U$ (logistic) & $d=3,5,7$ & $0.499,\,0.498,\,0.498$ (chance) \\
Adversary AUC with witness $W$ & $d=3,5,7$ & $0.982,\,0.986,\,0.990$ \\
Adversary AUC, coherent$+$leakage & $d=3$ (exact statevector) & $\in[0.497,0.507]$, five noise models \\
MMD two-sample test, honest / witness & $d=3,5$ & $p=0.45,0.76$ (n.s.) $/$ $p=0.003$ \\
Lattice-surgery error $p_{\rm LS}$ (seam) & $d=3\!\to\!11$ & $6.5\times10^{-3}\to2.2\times10^{-4}$, fit $e^{-0.43d}$ \\
\bottomrule
\end{tabular}
\caption{Simulation results. Rotated surface code, circuit-level depolarizing noise (one- and two-qubit
depolarization, measurement and reset flips, per-round data depolarization) at physical rate $p=0.005$ for the
classifier rows and $p=0.003$ for lattice surgery; MWPM decoding. The classifier rows use $40\,000$ shots per point ($20\,000$ per logical class), together with a stratified $70/30$ train/test split. AUC is evaluated on the held-out $30\%$, and $95\%$ CIs come from $1000$ bootstrap resamples of the test set. The coherent/leakage row uses exact statevector trajectories with $8000$ shots per point. The lattice-surgery row uses $40\,000$ shots per distance. Per-point manifest with seeds and tool versions in
Supplementary Note~20.5; extended checks in Supplementary Notes~8 and~20. \textbf{The first row is a consistency
check, not a measurement}. $\Xbar$ commutes with every measured stabilizer, so under Pauli noise the two detector laws are identical by construction and no classifier could report other than chance.}
\label{tab:numerics}
\end{table}

\begin{figure}[t]
\centering
\includegraphics[width=\linewidth]{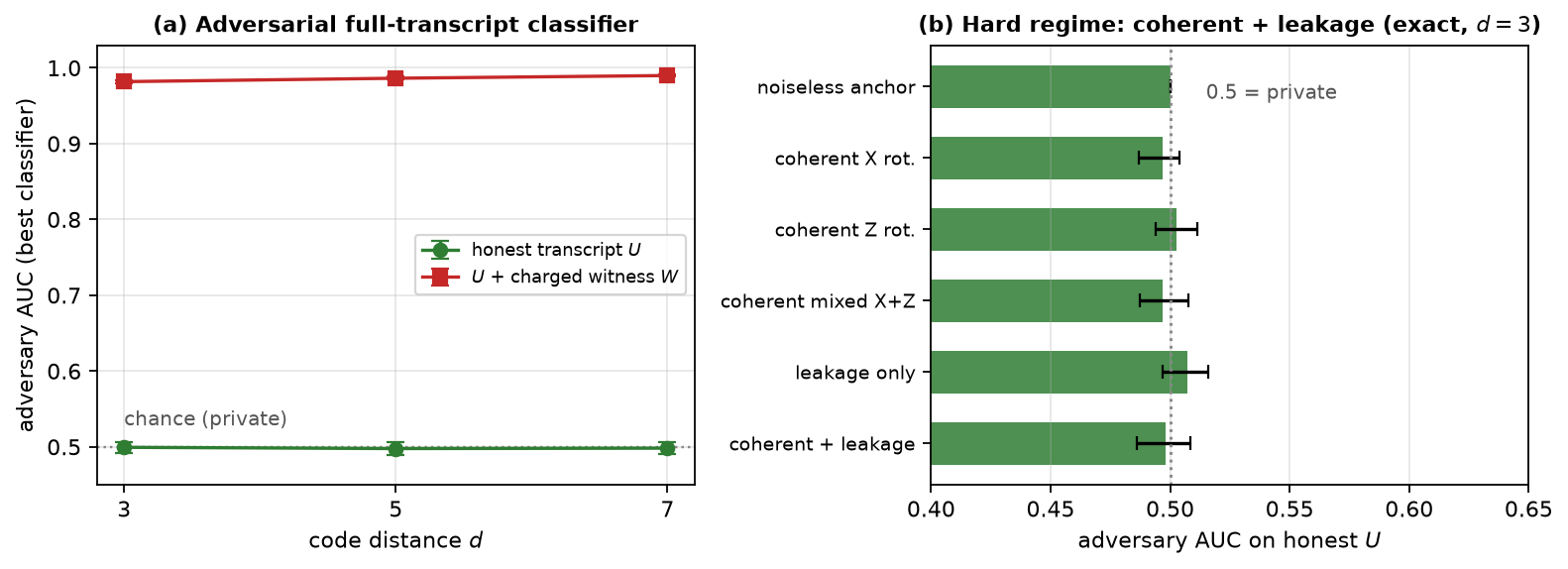}
\caption{An adversarial null for the honest transcript (bootstrap $95\%$ CIs; circuit-level depolarizing noise at
$p=0.005$, minimum-weight perfect-matching decoder, logistic adversary on the full detector record, held-out split).
(a) An adversary given the full detector record stays at chance for honest $U$ across $d=3,5,7$ and reaches near one
with a charged witness $W$. (b) The same adversary stays at chance under coherent and leakage noise in an exact
statevector simulation at $d=3$ (five noise models).}
\label{fig:nogo}
\end{figure}

The check gives the adversary the full record instead of a chosen statistic and complements the trained classifier with a kernel two-sample test (maximum mean discrepancy, characteristic kernel). The test probes equality of the transcript laws instead of the skill of a classifier family, and given enough samples it consistently detects any difference. Across $d=3,5,7$ the held-out area under the ROC curve sits at chance, $0.499$ to $0.498$. The permutation $p$-value is $0.45$ at $d=3$ and $0.76$ at $d=5$, and the model-free total variation of a coarse-grained statistic sits at its finite-sample floor. Appending a charged calibrated witness lifts the same adversary to $0.99$, the $p$-value to $0.003$, and the model-free statistic an order of magnitude above the floor, which shows the converse operationally (\cref{fig:nogo}a). At these shot counts, the classifier test resolves an AUC excess of about $0.011$ at $80\%$ power, and the bootstrap upper limit on the honest record's excess is below $0.01$. What each honest test excludes is recorded instead of left implicit, so the chance-level results bound the realizable separation from above instead of merely failing to find it (Supplementary Note~20.6).

Coherent and leakage noise are treated in the stress tests below, since a stabilizer simulator cannot represent them
and the exact statevector trajectory simulator used there reaches only $d=3$.

\paragraph*{The converse, and the locality calibration.}
The leaking regime is equally testable, and the locality hypothesis that powers the no-leakage theorem can be measured directly. \Cref{fig:converse} reports both.

\begin{figure}[t]
\centering
\includegraphics[width=\linewidth]{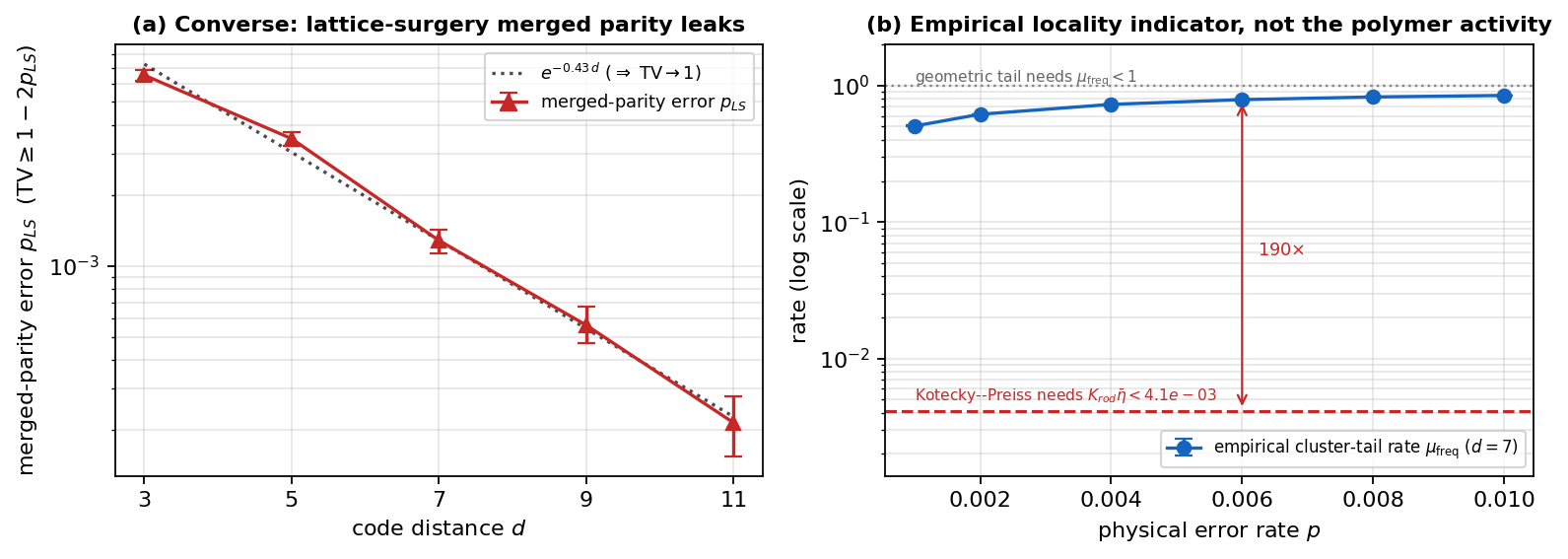}
\caption{The converse and the locality indicator (bootstrap $95\%$ CIs). (a) The lattice-surgery merged-parity error
$p_{LS}$ decays exponentially in $d$ (MWPM decoder), so the converse bound $1-2p_{LS}$ rises toward one. (b) The
\emph{empirical} cluster-tail rate $\mu_{\rm freq}$, fitted to the sizes of the observed fired-detector components
$\Robs$, which are not the fault-support regions $\Rlat$ the theorem is stated for. This rate also differs from the polymer activity. The Koteck\'y--Preiss condition $ea<1$ requires $K_{\rm rod}\bar\eta<1/(e^3\,\mathrm{deg})\approx
4.1\times10^{-3}$ at detector degree $12$, a threshold drawn as the dashed line, while the measured indicator sits about $190\times$ above it. A geometric tail with $\mu_{\rm freq}<1$ supports the \emph{form} PWTS assumes. It does not certify the smallness the theorem needs.}
\label{fig:converse}
\end{figure}

The converse is quantitative and scales as predicted. We model the lattice-surgery joint-parity measurement by its
distance-$d$-protected merge seam and decode it with minimum-weight perfect matching, a standard named decoder. The
merged-parity error $p_{\rm LS}$ falls geometrically with distance, from $6.5\times10^{-3}$ at $d=3$ to
$2.2\times10^{-4}$ at $d=11$, an exponential fit with decay rate $0.43$ per unit distance (\cref{fig:converse}a). The
converse total-variation bound $1-2p_{\rm LS}$ therefore rises toward one. The implication for a designer is direct. A logical measurement record becomes a more reliable leak as the distance grows, which is the opposite of the no-leakage theorem. This trend confirms that the boundary between the two sides of the device map is not an artifact of small distance.

The locality hypothesis has a measurable proxy, and how far the proxy goes must be stated exactly, once, since PWTS asks
that the pre-logical activity $\eta(R)$ of a detecting region decay geometrically in $|R|$. The data instead yield the tail rate $\mu_{\rm freq}$ of the size distribution of the \emph{observed components} $\Robs$, which are connected components of detectors that actually fired (\cref{fig:converse}b). This behavior supports the geometric \emph{form} PWTS assumes and its uniformity in $d$, but it does not certify the smallness the closed bound needs. The rate stays below one at every distance tested and converges as the patch grows. Two recorded gaps separate the fitted rate from the rate PWTS assumes. The first is that the fit counts how often components of a given size occur, while the hypothesis bounds the resummed operator-norm activity such a region can carry. The second is that the fit sees the fired-detector component $\Robs$, while the hypothesis is stated over the latent fault-support region $\Rlat$. It is measured on $\Robs$ instead of on the $\Rlat$ of \Cref{lem:cosetweight}, which a cancelling fault chain can make arbitrarily smaller than its own footprint. The Koteck\'y--Preiss condition $ea<1$ requires $K_{\rm rod}\bar\eta<1/(e^3\,\mathrm{deg})\approx4.1\times10^{-3}$ at detector degree $12$, and the measured indicator sits about $190\times$ above it at the simulated rates. On hardware calibration, the corresponding gap is $21.5\times$ in $\varepsilon_0$. \emph{This is the screening/certification distinction referred to throughout}. Passing the diagnostic does not license \Cref{thm:main}, which also needs the transcript factorization, the microscopic channel-norm bounds, and the smallness condition. 
\paragraph*{Stress tests.} These split in two. Where the noise is Pauli, the $\Xbar$-prepared pair has identical detector laws by covariance, so varying the classifier or the Pauli noise model cannot change the answer. The tested classifiers are gradient-boosted trees and a multilayer perceptron, while the tested variations include $Z$-biased single-qubit noise and finite-range $Z$--$Z$ crosstalk, with all tests returning chance across $d=3,5,7$. The coherent and leakage tests differ in kind and provide the only simulation evidence here that could have come out otherwise because a coherent over-rotation channel is not $\Xbar$-covariant and the algebra permits a leak. Under coherent $X$, coherent $Z$, mixed coherent, pure
leakage, and combined coherent-plus-leakage noise the adversary stays within statistical distance of one half
(\cref{fig:nogo}b), with a noiseless anchor of exactly one half. That is a genuine null at $d=3$ only, because the exact statevector simulator does not reach further, but it is not a large-distance coherent-noise scaling test, and we do not present it as one. Finally, two statements about the physical error rate must be kept apart. The specific honest memory stays empirically blind across a wide sweep even above the code threshold, which is a symmetry consequence and not an extension of the theorem. By contrast, the PWTS cluster-tail rate that the proof needs climbs toward one as the rate rises, but we do
not claim the privacy and code thresholds coincide. The converse is confirmed a second way with two independent distance-$d$ patches decoded separately. Their joint logical parity decays as $e^{-0.44d}$ from $1.3\times10^{-2}$ at $d=3$ to $3.9\times10^{-4}$ at $d=11$. This independent-patch quantity is about $2p_L$, and we do not claim that it bounds the merged-seam error. Extended checks are in Supplementary Notes~8 and~20.

\subsection*{A diagnostic test of when the theorem applies}
The no-leakage theorem and the converse together give a usable screen for whether the theorem can be invoked for a given device, and this screen
has two parts that match the two hypotheses and certify different things.

\emph{Part one, exact.} Compute the group generated by the measured checks in binary-symplectic form and verify that
it contains no nontrivial logical operator. This runs in polynomial time and it certifies, exactly, that the idealized
measured algebra carries no explicit logical witness, and the honest rotated surface-code gadget passes at $d=3,5,7$. A lattice-surgery instrument fails because $\Zbar_1\Zbar_2$ lies in the measured group, and the witness it returns is precisely the record whose leak the converse quantifies. It cannot reach coherent faults, analog records, leakage,
timing or feedback, which is why part two exists.

\emph{Part two, empirical.} Estimate the transcript's locality from the connected-component size distribution of fired detectors, whose geometric tail rate $\mu_{\rm freq}$ is fitted with confidence intervals (\cref{fig:converse}b). The controlled alternative toggles a local Pauli and measures the response against region size. The measurement shown here is \emph{passive}, whereas that controlled alternative is a different experiment described in Supplementary Note~11 and is not the one shown here. The fitted locality rate therefore keeps the status established above. Full certification would also require a
device model showing that every exposed analog and classical record factors through allowed local detector events or
input-independent kernels.

The order of the two parts matters in practice. The subtler failures include a persistent sensor or a long-memory record. The algebra check is cheap and catches deliberate logical channels immediately, while the calibration guards against those subtler failures that no static algebra check could see.

\section*{Hardware measurement}
\label{sec:hardware}

The preceding sections establish a conditional theorem. This section reports what a real superconducting processor
does, and the answer is that our sufficient certificate for it is not satisfied there, by a large and measurable
margin. All data below were taken on
\texttt{ibm\_cleveland}, a 156-qubit IBM Heron~r2 device, in \SI{24}{\minute} of quantum-processor time charged
across eight jobs. Raw per-shot records, per-qubit calibration snapshots taken both at submission and after each job,
job identifiers, the realised acquisition order and the analysis code are in the repository.

\paragraph{Acquisition artifact is bounded from above, on the device.}
Every number below is read against a control acquired on the same qubits, in the same job, under the same randomised interleaving. The control prepares one logical class twice and hands it to the estimator under two different labels, so any total variation present there reflects drift or acquisition artifact and does not reflect logical-state dependence.

Bounding that requires an inequality pointing the opposite way from the one used everywhere else in this work. The selection-free statistic is a \emph{lower} bound, and a zero reports only that one classifier at one fixed threshold found nothing. This is a failure to detect, not a demonstration of absence, and the same point applies equally to the randomized-encoding arm discussed below, which is why we do not quote the null cells' lower bounds. What those cells admit, and the leak cells do not, is a genuine upper bound because they run the shortest circuits in the session, so their entire record is four or six bits. On an alphabet of $k\in\{16,64\}$ categories with $n\ge20\,000$ shots per class, the empirical laws pin down the true ones. Jointly with probability at least $95\%$, the distribution-free Bretagnolle--Huber--Carol inequality gives
\begin{equation}
\TV(P,Q)\;\le\;\TV(\widehat P,\widehat Q)+\!\!\sum_{i\in\{0,1\}}\!\!\sqrt{\frac{k\ln 2+\ln(4/\delta)}{2n_i}}.
\label{eq:nullupper}
\end{equation}
It assumes nothing about agreement between the two classes, because their agreement is exactly what the test evaluates. The four null cells were run on the repetition line at the shortest and longest exposures and on the $[[4,1,2]]$ layout. Across these cells, the empirical total variations are $0.002$ to $0.011$, and \cref{eq:nullupper} bounds every one of them above by $0.051$. The like-for-like leak cell, same line and same
exposure, has a $95\%$ \emph{lower} bound of $0.267$. The two intervals do not meet, so two inequalities pointing the right ways establish the separation between a leak and an acquisition artifact on the device without a noise-model inference.

\paragraph{What is measured, and on which code.}
The leak experiment runs a \emph{repetition-code} $Z$-memory, not the rotated surface code of \Cref{thm:main}. This is deliberate and is the point of the design because a repetition code has $Z$-distance $d_Z=1$, with a logical basis given by the physical computational basis in which amplitude damping acts. It is therefore the least protected case for damping-induced charged leakage, and the cheapest setting in which to
resolve the $d_Z=1$ mechanism. It is \emph{not} a worst case for backbone blindness because, under the convention fixed below, H1$^\ast$ remains intact and charged fault-path contributions carry the input dependence.
The surface code, with $d_Z=d$, is the other end of the same criterion and is measured separately below, so readers should not read the repetition-code numbers as a test of \Cref{thm:main}. The repetition-code experiment tests the predicted $d_Z=1$ leakage mechanism, while the separate $21.5\times$ certificate calculation assesses whether the theorem's sufficient smallness condition can be invoked on this device. Neither is evidence for the other.

\begin{figure}[t]
\centering
\includegraphics[width=\linewidth]{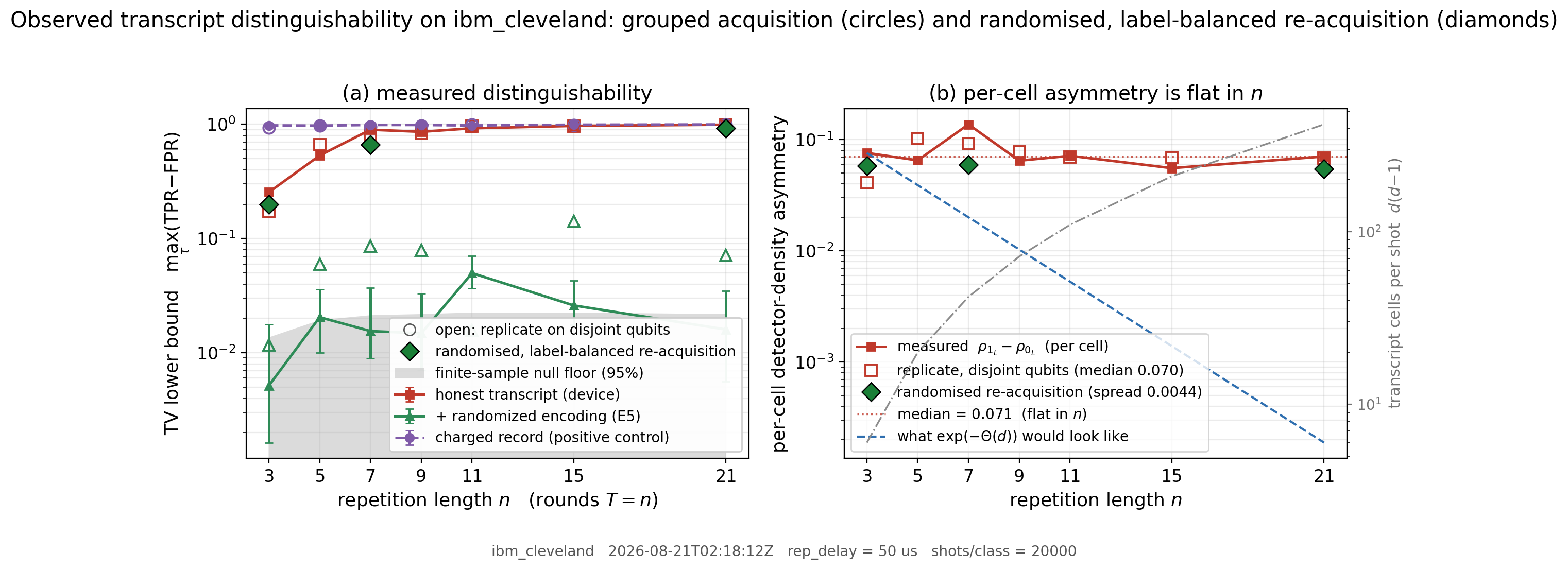}
\caption{Measured transcript privacy on \texttt{ibm\_cleveland}. Filled circles are the original grouped acquisition. The \textbf{filled diamonds are the randomised, label-balanced re-acquisition} at $n=3,7,21$, which is the acquisition used for the headline number. (a) Total-variation lower bound against code repetition length for the honest transcript, randomized encoding, and a deliberately planted charged record. The shaded band is the finite-sample null floor, and open markers repeat the whole sweep on a disjoint qubit set. Randomisation moves $n=21$ from $0.949$ to $0.927$ and $n=7$ from $0.880$ to $0.660$, and the second change provides the control that detects acquisition-order contamination without ratifying the original value. Repetition has $d_Z=1$, so
\Cref{prop:attained} gives a first-order leak that lengthening the code does not suppress. (b) The per-cell detector-density asymmetry is flat in $n$ as that mechanism requires and flatter still under randomisation (spread $0.0044$ against $0.067$). The dashed curve shows what an exponentially suppressed leak would look like, and the grey dash-dotted curve is the number of transcript cells per shot, $n(n-1)$.}
\label{fig:hw_e1}
\end{figure}

\begin{figure}[t]
\centering
\includegraphics[width=0.56\linewidth]{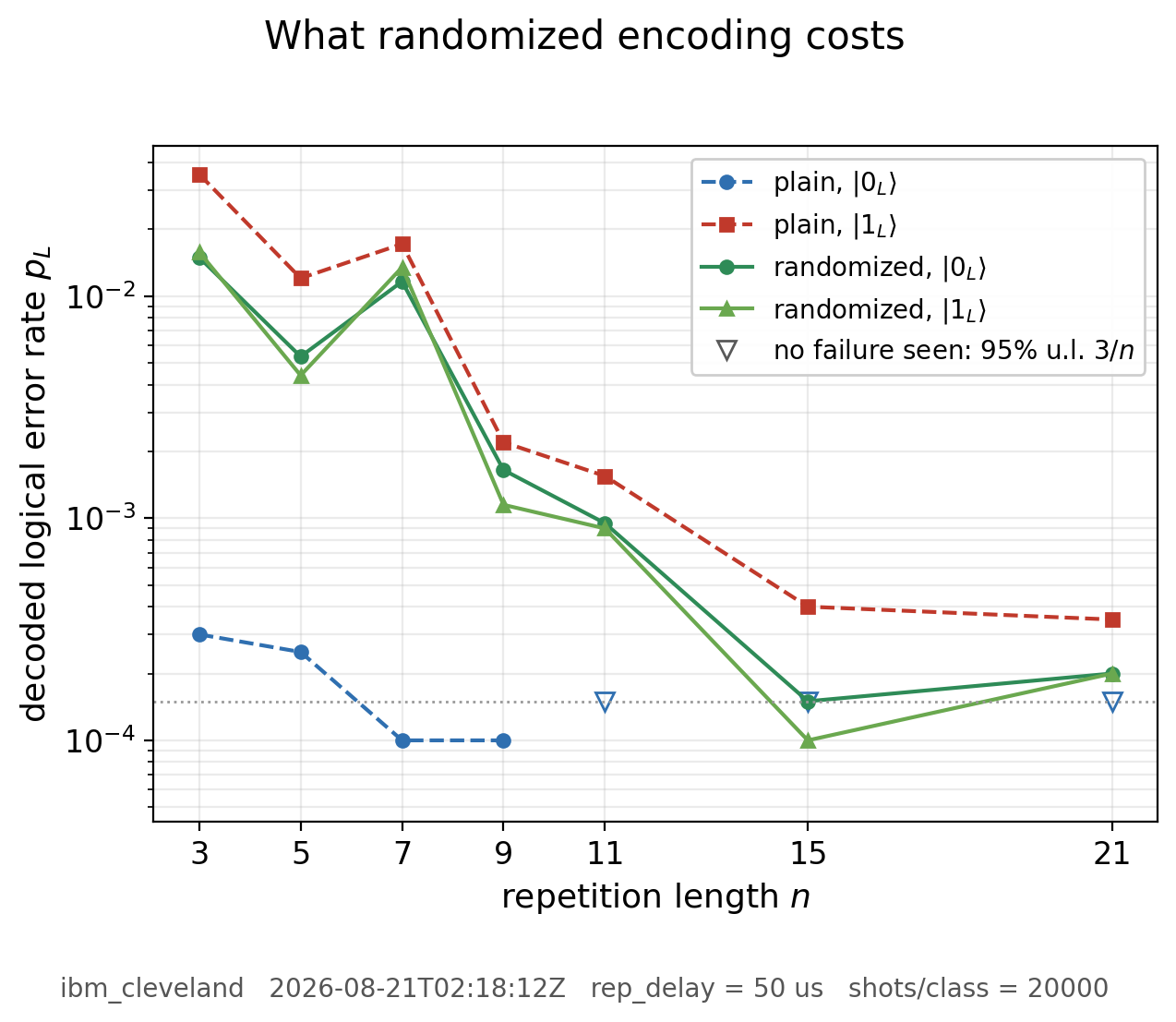}
\caption{What the mitigation costs. Decoded logical error rate per logical input, plain and with randomized encoding.
Randomized encoding equalises the two inputs by moving both to their mean, so the cheaper input pays. At $n=3$, the plain memory fails at $3.0\times10^{-4}$ for $\zL$ against $3.5\times10^{-2}$ for $\oL$, and after randomized encoding the two agree to within $6\%$. Open triangles are rule-of-three $95\%$ upper limits ($3/n$) at points where no logical
failure was observed in $20\,000$ shots.}
\label{fig:hw_cost}
\end{figure}

\subsection*{The honest transcript is distinguishable, under randomised acquisition}

For each repetition length $n\in\{3,\dots,21\}$ we run $T=n$ rounds and prepare each logical basis state $20\,000$ times,
withholding the final data readout from the adversary. The reported statistic is the total-variation lower bound $\max_\tau(\mathrm{TPR}-\mathrm{FPR})$ on a held-out split, which is the quantity \Cref{thm:main} bounds from above. Significance is evaluated by test-label permutation ($10^3$ draws, Bonferroni-corrected across model families).

That statistic maximises over $\tau$ on the split that then scores it, so as a point estimate it is optimistically
biased. Freedom from threshold and model-selection bias is a weaker and different claim than causal attribution to the logical input. We therefore report a second, \emph{selection-free} quantity as the headline, computed from the same shots by a three-way split. We fit the score on $50\%$, choose \emph{both} the model family and the single threshold $\tau^\ast$ on a disjoint $25\%$, and evaluate $\mathrm{TPR}-\mathrm{FPR}$ at that fixed $\tau^\ast$ on the remaining $25\%$, never re-optimised. With $\tau^\ast$ fixed before the test split is touched, $\mathrm{TPR}$ and $\mathrm{FPR}$ are computed on disjoint groups of shots, but disjoint is not the same as independent, and the acquisition order matters here because the circuits were \emph{not} randomly interleaved. For each length, the job ran the $b=0$ circuits of all three arms and then the $b=1$ circuits, with each circuit's shots contiguous, which confounds the logical label with position in the job. Two controls in the existing data bear on whether that confound produces the signal, and a third
measurement, reported immediately below, settles it.

\emph{Control 1: the mitigated arm spans the same gap.} The randomized-encoding circuits carry the same $b$ label,
sit adjacent to the plain circuits in the same job, and their $b=0$ and $b=1$ halves are separated by the \emph{same}
temporal gap. Any drift acting across that gap would act on both arms. Measured with the same estimator, the
randomized-encoding arm returns a point estimate of $|{\le}0.0006|$ and a bound of $0.000$ at every length,
against $0.949$ for the plain arm.

\emph{Control 2: contiguous shot splits.} Splitting shots at random lets training and test share a
drift state. The procedure trains on the first half of each circuit's shots and chooses $\tau^\ast$ on the third quarter, and it evaluates on the last quarter, which contains the latest data and was not used earlier. Re-running the three-way estimate on \emph{contiguous} blocks gives $0.958$ at $n=21$ and $0.218$ to $0.958$ across the sweep, in line with the random-split values (\texttt{sim/hw/acquisition\_controls.py}).

\emph{Control 3: the randomised re-acquisition.} Neither control above excludes a drift aligned with the $b=0/b=1$ boundary that is also suppressed under randomized encoding. Only randomised acquisition does. We therefore re-ran the three representative lengths with the circuit order randomly permuted under a recorded seed. The logical label was balanced across the job ($0.50$ in each half), and the re-acquisition cost \SI{45}{\second} of processor time (\texttt{sim/hw/submit\_e1.py --interleave-seed}). With the labels randomly assigned, device drift cannot be
systematically bound to the logical input.

\begin{center}\small
\begin{tabular}{lccc}
\toprule
one-sided $95\%$ lower bound & $n=3$ & $n=7$ & $n=21$\\
\midrule
grouped acquisition (all $b{=}0$, then all $b{=}1$) & $0.239$ & $0.880$ & $0.949$\\
\textbf{randomised, label-balanced acquisition} & $\mathbf{0.199}$ & $\mathbf{0.660}$ & $\mathbf{0.927}$\\
\addlinespace
per-cell asymmetry, grouped & $0.076$ & $0.138$ & $0.071$\\
per-cell asymmetry, randomised & $0.058$ & $0.059$ & $0.055$\\
\bottomrule
\end{tabular}
\end{center}

Three things follow. The headline survives because at $n=21$ the bound moves from $0.949$ to $0.927$, establishing the separation as a property of the logical input independent of shot timing. The randomized-encoding arm again returns $0.000$ at every length, but in that run the $n=7$ cell was partly an artefact, as shown by $0.880\to0.660$, with its per-cell asymmetry falling from $0.138$ to $0.059$. That cell was the one outlier in the grouped sweep, sitting well above the others with no explanation. Acquisition order explains the outlier. The cell remains in the report, and we note that its consequence is helpful, not harmful. Under randomised acquisition, the per-cell asymmetry is flat at $0.055$--$0.059$ across a sevenfold range of $n$, considerably flatter than the grouped $0.071$--$0.138$. Flatness in
$n$ is what $d_Z=1$ predicts, and removing the acquisition-order contamination makes the prediction \emph{cleaner},
not weaker.

The remaining four lengths of the sweep ($n=5,9,11,15$) were acquired in the original grouped order and are reported as such. The randomised control covers $n=3,7,21$. We take an exact one-sided Clopper--Pearson \emph{lower} bound on $\mathrm{TPR}$ and an exact one-sided \emph{upper} bound on $\mathrm{FPR}$, each at level $\alpha/2=0.025$. By the union bound, their difference is a one-sided lower bound with joint coverage of at least $95\%$ (\texttt{sim/hw/threeway\_tv.py}), which exchangeability alone cannot guarantee because correlated exchangeable draws can be overdispersed and break the coverage. That coverage statement is conditional on the shots within a class being independent and identically distributed Bernoulli draws, and the controls above are what we can offer in its support. Under that model the number is a lower bound on the population total variation at the stated
confidence, not a selected empirical statistic.

\begin{table}[htbp]
\centering\small
\begin{tabular}{lccccccc}
\toprule
repetition length $n$ & 3 & 5 & 7 & 9 & 11 & 15 & 21\\
\midrule
honest transcript & 0.254 & 0.533 & 0.897 & 0.860 & 0.923 & 0.968 & \textbf{0.991}\\
\quad replicate, disjoint qubits & 0.172 & 0.665 & 0.811 & 0.837 & 0.960 & 0.968 & \textbf{0.993}\\
$+$ randomized encoding & 0.005 & 0.021 & 0.016 & 0.015 & 0.050 & 0.026 & 0.016\\
charged record (positive control) & 0.978 & 0.974 & 0.986 & 0.986 & 0.976 & 0.991 & 0.996\\
\midrule
\multicolumn{8}{@{}l}{\emph{selection-free lower bound, grouped acquisition}}\\
honest transcript & 0.239 & 0.496 & 0.880 & 0.807 & 0.882 & 0.914 & \textbf{0.949}\\
$+$ randomized encoding$^\dagger$ & 0.000 & 0.000 & 0.000 & 0.000 & 0.000 & 0.000 & 0.000\\
\midrule
per-cell asymmetry $\rho_{1_L}-\rho_{0_L}$ & 0.076 & 0.065 & 0.138 & 0.065 & 0.072 & 0.056 & 0.071\\
\quad replicate & 0.041 & 0.103 & 0.092 & 0.078 & 0.070 & 0.069 & 0.069\\
\bottomrule
\end{tabular}
\caption{Measured transcript distinguishability on \texttt{ibm\_cleveland}. Total-variation lower bounds,
$20\,000$ shots per logical class per entry, $T=n$ rounds, inter-shot delay \SI{50}{\micro\second}. The $^\dagger$A lower bound of $0.000$ asserts only $\TV\ge0$. It is consistent with the by-construction equality of the two laws but does not certify it, since the estimator returns no upper bound. The
finite-sample null floor is $\approx0.022$ and every honest-transcript and positive-control entry has permutation
$p=0.001$. ``Replicate'' repeats the entire sweep on a qubit set disjoint from the first. The rows headed
\emph{selection-free lower bound} are the one-sided $95\%$ bound from a three-way split, with $\tau$ fixed
before the test split is touched and coverage taken under the within-class binomial model. Those rows come
from the grouped acquisition, and the randomised control above covers $n=3$, $7$ and $21$.}
\label{tab:hw}
\end{table}

Three things are observed at once, and they correspond to the three arms of a design whose analysis plan,
circuit families and result manifests were fixed before any machine time was spent on \emph{this} sweep. The randomised re-acquisition and the exposure sweep that follow are declared as follow-ups, each frozen before its own job was submitted.

\emph{The honest transcript is distinguishable, and adding repetition does not help.} At repetition length $n=21$, an adversary reading only the syndrome record attains $\TV\ge0.927$ at one-sided $95\%$ confidence under randomised
acquisition (\cref{fig:hw_e1}). This corresponds to a success probability of at least $0.963$ and an advantage over guessing of at least $0.463$, each truncated downward so that it remains a lower bound. (The selected statistic on that same randomised re-acquisition is $0.940$, a gap of $0.013$, while on the earlier grouped sweep, the corresponding pair is $0.991$ against $0.949$, a gap of $0.042$. Both gaps are measurements of selection bias, not matters of argument, and the gap is smaller under randomised acquisition.) This is what \Cref{prop:attained} gives, not a counterexample to \Cref{thm:main}, because the repetition code has $d_Z=1$ at every $n$ and its first-order leak in $\gamma$ cannot be suppressed by lengthening the code. Adding qubits buys correction against bit flips and buys nothing at all against transcript exposure, and that is the prediction being tested, while the observation is consistent with it.

\emph{The rise is not an artifact of the classifier or of the qubits.} The same analysis chain returns the null floor when the leak is absent and returns $\approx0.98$ when a charged record is deliberately planted, at every distance. The whole sweep reproduces on a disjoint qubit set, landing at $0.993$ against $0.991$.

\emph{Randomized encoding removes it, at no two-qubit-gate cost, for the pair demonstrated.} On each shot we prepare
the codeword $b\oplus r$ for a secret random bit $r$ and relabel the outcome by $r$ afterwards, so the transcript law
is $(P_0+P_1)/2$ for either \emph{basis} input. Measured, this returns the empirical lower bound to the finite-sample
floor at six of seven lengths, a $62$-fold reduction at $n=21$.

The privacy argument for this mitigation is a construction, not a measurement, and should be read that way. It also
carries a trust boundary, which we state explicitly because the threat model elsewhere in this paper grants the
observer full circuit and device knowledge. The provider can know the entire circuit that produces the randomized codeword. The secret bit $r$ is generated by the client or a trusted input-preparation agent and never leaves that party, while the provider receives the preparation of $c=b\oplus r$. The provider can know $c$ while knowing neither $b$ nor $r$, and the client alone undoes the relabelling on the returned outcome. The randomization is
therefore not a layer the provider applies and could read off, which is precisely why it must be physical and why the
key must not appear anywhere in the transcript. Given that boundary, the two basis inputs produce the same observed law. If $r$ is uniform, is not exposed anywhere in the transcript, and the device physically prepares $c=b\oplus r$, then for either basis input the observed law is $P(U\mid b)=\tfrac12\bigl(P(U\mid 0_L)+P(U\mid 1_L)\bigr)$, the same object for $b=0$ and $b=1$. The two laws are
therefore \emph{identical by construction} and $\TV=0$ exactly, under that threat model. The implementation check asks whether the randomization was really applied to the physical state and not to a label, and the hardware numbers provide this check, not independent evidence of privacy. A one-sided lower confidence bound of zero asserts only $\TV\ge0$ and certifies nothing, and the three-way procedure returns this bound for the arm at every length. It is consistent with equality but does not establish it because establishing it empirically would require an \emph{upper} bound on $\TV$, which the estimator does not provide.

First, the cost. This is a \emph{physical} randomization, not a relabelling. Merely changing a Pauli-frame label in software leaves the hardware in the same state and the environment damps it identically, so the transcript would not be twirled at all. Preparing the opposite codeword costs a transversal layer of single-qubit $X$ gates, and on the transpiled circuits at $d=5$, the measured changes are $+5$ one-qubit gates and $+1$ depth. They include \textbf{zero} additional two-qubit gates.
The claim ``No two-qubit-gate cost'' is accurate. The claim ``zero gate cost'' is inaccurate.

Second, an $\Xbar$ randomization acts trivially on $\lvert\pm_L\rangle$, since
$\Xbar\lvert\pm_L\rangle=\pm\lvert\pm_L\rangle$, so it symmetrises the pair $\{\zL,\oL\}$ and covers the
computational-basis label. Covering an unknown state needs the full logical Pauli
twirl over $\{I,\Xbar,\bar Y,\Zbar\}$, again applied physically, with the key unknown to the observer and tracked
through everything downstream. The gate cost reported above is that of basis-state randomized encoding in a memory experiment, measured on
the pair it covers.

\paragraph{The per-cell asymmetry, which is flat in $n$.}
A rising total variation invites an objection based on transcript length because the record grows as $n(n-1)$ cells per shot, from $6$ at $n=3$ to $420$ at $n=21$. The sharper quantity is the \emph{per-cell} asymmetry. Write $\rho_b$ for the fraction of the $n(n-1)$ recorded detector cells that fire, averaged over the shots of logical class $b$. The per-cell asymmetry is $\rho_{1_L}-\rho_{0_L}$, shown in the last two rows of \Cref{tab:hw}, and it is flat across a sevenfold range of $n$, with a median of $0.071$ and a value of $0.070$ on the independent replicate.

What this does and does not measure needs care, because the limits differ. \Cref{prop:dzlaw,prop:attained} concern a
\emph{fixed} code at \emph{fixed} $T$ as $\gamma\to0$. The sweep here holds the device and its noise essentially fixed while the circuit grows with $T=n$, so it
measures how distinguishability accumulates with record length at fixed $\gamma$. The exponent in $\gamma$ is
measured separately, by the fixed-code exposure sweep below. On this processor, the sweep establishes nonvanishing per-cell asymmetry for the $d_Z=1$ repetition family as the code lengthens, which is consistent with a first-order damping mechanism whose per-round, per-qubit coefficient does not fall with $n$. Total distinguishability then accumulates over a growing record instead of decaying, but consistency is not a measured slope.

The natural next experiment directly tests the order with a fixed code and effective damping exposure varied through each round's idle duration, followed by a fit of $\mathrm{d}\log\widehat{\TV}/\mathrm{d}\log\gamma$. The prediction is a slope near $1$ for a $d_Z=1$ repetition memory and near $3$ for the $d_Z=3$ surface patch, and we have verified exactly that behaviour in simulation, with measured slopes $0.9995$, $1.9994$ and $2.9995$ against $d_Z=1,2,3$. The hardware counterpart of that exponent measurement is the fixed-code exposure sweep reported below.

This measures the code, not the hypotheses. It is not by itself evidence that H1$^\ast$ or PWTS fail on this device. Under the convention used throughout (below), the non-identity component of the no-jump branch is a fault atom, not part of the honest backbone. Thus H1$^\ast$ can hold while charged contributions carry the entire leak at order $\gamma^{d_Z}$. What the hardware establishes is that on a real processor the damping asymmetry is large enough
to be resolved with $20\,000$ shots, and that for $d_Z=1$ it is not suppressed. The separate question of whether the
theorem's hypotheses hold here is answered by the certificate below, which reports a $21.5\times$ shortfall.

\paragraph{A convention that must be fixed.}
Amplitude damping admits two decompositions and they lead to opposite-sounding statements. The first decomposition places the entire no-jump Kraus branch in the honest backbone, where that branch carries a $Z$-dependent component and H1$^\ast$ fails outright. The second keeps only the scalar part in the backbone and treats the non-identity $Z$ component of the no-jump branch as a local fault atom. Under this decomposition, H1$^\ast$ holds and the input dependence appears as charged fault-path contributions. \textbf{We use the second convention throughout} because the PWTS derivation and the $O(\gamma^{d_Z})$ analysis both assume it, and under this convention, amplitude damping is a non-blind local fault covered by the theorem, not an assumption violation excluded by it. Under this convention H1$^\ast$ remains intact throughout, and charged fault-path contributions account for the observed asymmetry without any failure of backbone blindness.

\paragraph{Where the asymmetry comes from.}
On this processor, a syndrome round lasts \SI{5.6}{\micro\second}. Ancilla reset and measurement take \SI{2.68}{\micro\second} and \SI{2.65}{\micro\second}, respectively, while two-qubit gates take \SI{0.14}{\micro\second}. The data qubits are idle for the measurement and reset windows. The two-qubit layers occupy $\SI{0.14}{\micro\second}$ of the $\SI{5.47}{\micro\second}$ round, so at first order in $\gamma$ they carry about $2.5\%$ of the damping exposure and the idle windows carry the rest. The relevant feature is how long the data qubits wait on readout. The exposure the sweep varies is thus a property of the syndrome-extraction \emph{cycle} and not of the gates, and it is the quantity a faster or non-demolition readout would change.

\subsection*{The leak order is the code's $Z$-distance}
The two codes are compared in \cref{fig:hw_e2}.

The same device, in the same session and at the same physical error rates, was given both a three-qubit repetition code and a distance-3 rotated surface code for both logical inputs.

\begin{center}\small
\begin{tabular}{lccc}
\toprule
code & $d_Z$ & empirical TV \emph{lower bound} & verdict\\
\midrule
repetition, $n=3$ & 1 & $\ge0.358$ & leaks ($p=0.001$)\\
rotated surface, $d=3$ & 3 & $\ge0.032$ & not resolved above the floor $0.023$\\
\bottomrule
\end{tabular}
\end{center}

The tested statistic yields an empirical TV lower bound $11.2$ times smaller for the $d_Z=3$ circuit than for the
$d_Z=1$ circuit. That phrasing is deliberate and the stronger one is not available because both entries are \emph{lower} bounds, so they do not determine the ratio $\TV_{\rm rep}/\TV_{\rm surf}$. The bounds $\TV_{\rm rep}\ge0.358$ and $\TV_{\rm surf}\ge0.032$ allow the surface code's true total variation to exceed $0.032$ without resolution by this statistic. The verdict follows a stated rule. A null floor is the $95$th percentile of the statistic over $1000$ random permutations of the test labels, which is the value label-shuffled data reaches by chance alone, and the quoted intervals are bootstrap intervals over the test split. The surface-code bound is $0.032$ with a $95\%$ interval $[0.016,0.051]$, and that interval straddles the null floor $0.023$, so the bound is not resolved above it. The repetition bound $0.358$ has interval $[0.346,0.373]$, which clears its floor $0.016$ entirely. A lower bound sitting at the finite-sample floor reports a failure to detect, not an upper bound on what is there.

Nor can the routing overhead be used to rescue the claim. The surface-code circuit is routing-dominated on a heavy-hex map, using $1125$ two-qubit gates against the repetition code's $48$, but extra noise does not necessarily increase distinguishability. Additional state-independent error randomises both classes together and can reduce it, so we report the comparison as a statistic-specific observation consistent with \cref{eq:dzlaw}, and not as a measured suppression of the leak. The needed bound can come from an exact likelihood ratio or a calibrated generative model, with a sufficient statistic admitting two-sided bounds as another option. Establishing a suppression would require a confidence upper bound on $\TV_{\rm surf}$, which we do not have.

\begin{figure}[t]
\centering
\includegraphics[width=0.62\linewidth]{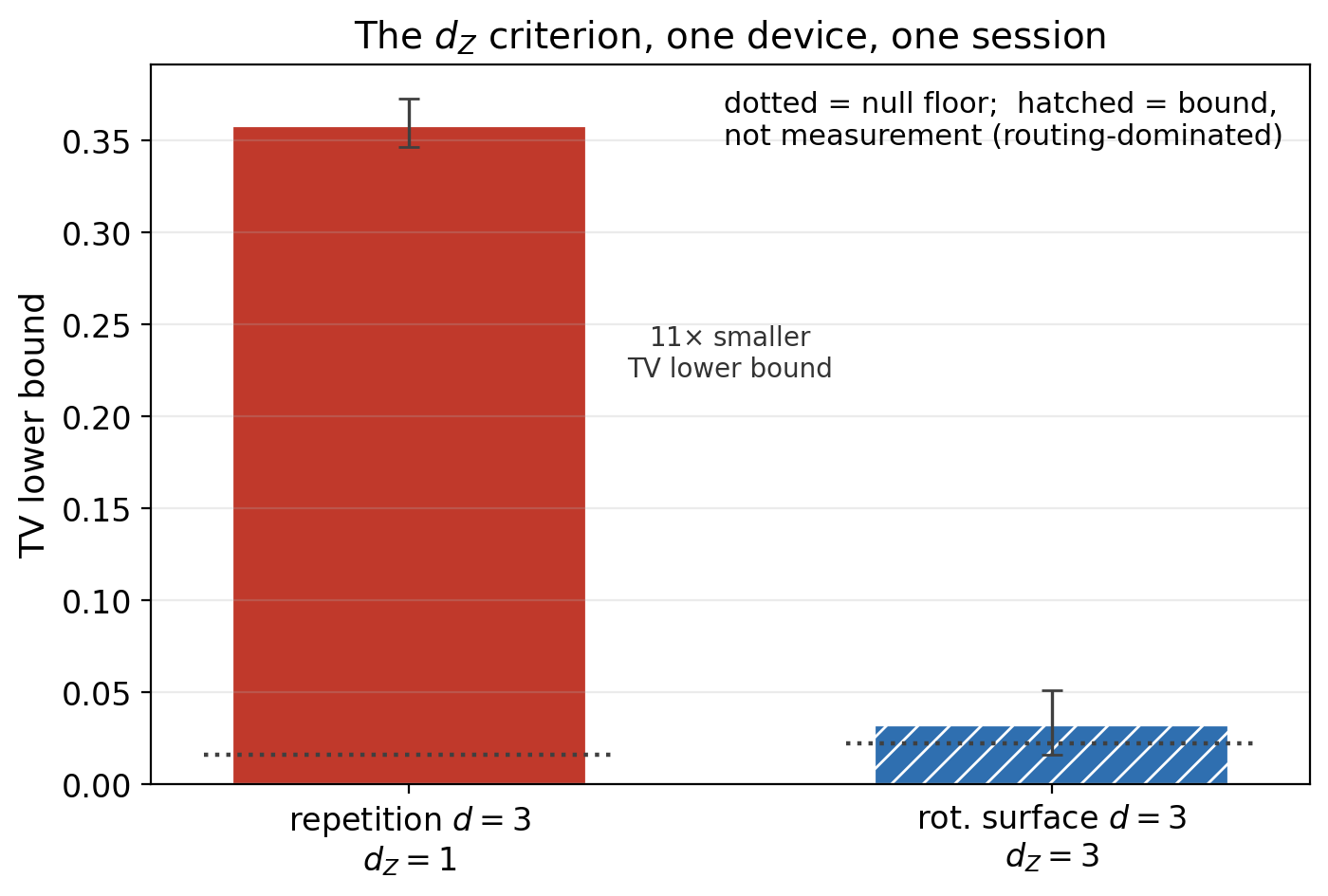}
\caption{The $d_Z$ criterion, measured. A $d_Z=1$ code ($n=3$ repetition) and a $d_Z=3$ code (distance-3 rotated surface) were measured on the same device in the same session. Dotted lines are null floors, and the hatched bar marks the surface code as a bound instead of a measurement because its circuit is routing-dominated on heavy-hex.}
\label{fig:hw_e2}
\end{figure}

\begin{figure}[t]
\centering
\includegraphics[width=\linewidth]{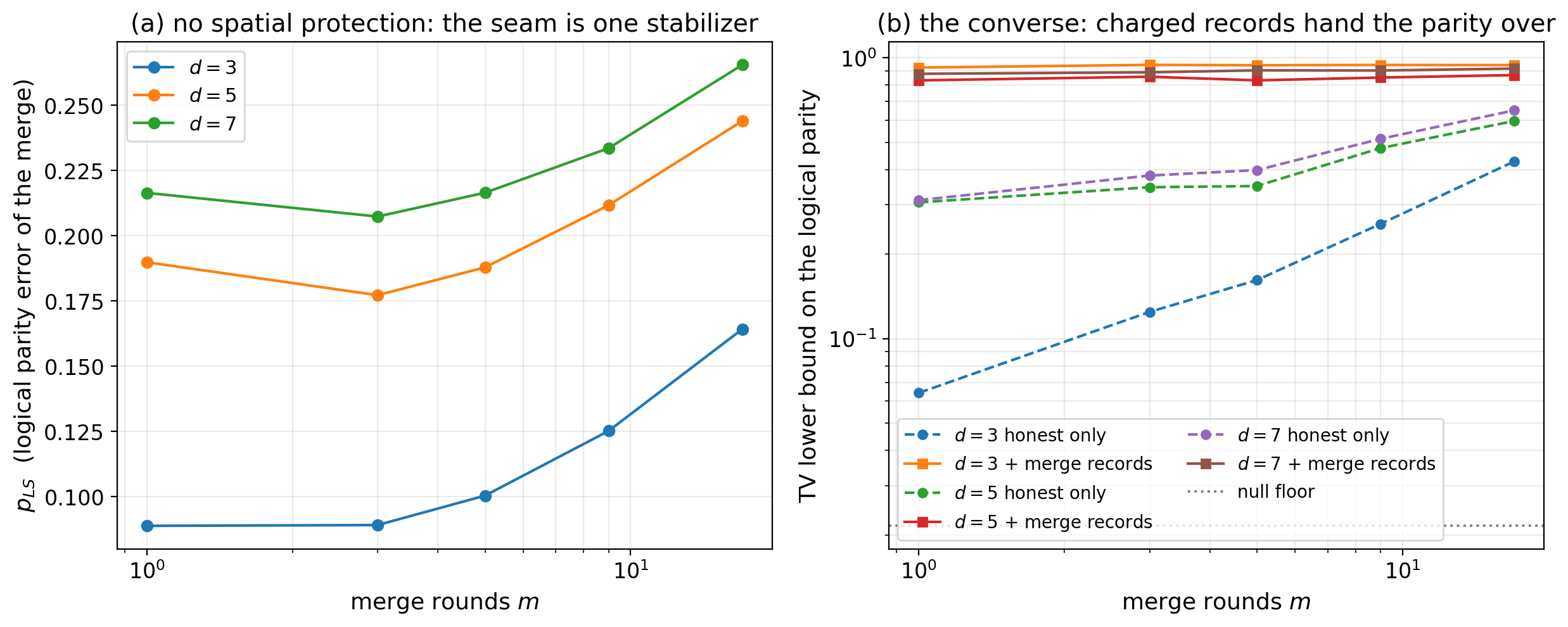}
\caption{The converse, measured. (a) For $d=3,5,7$, merged-parity error is plotted against merge durations of $1,3,5,9$ and $17$ rounds on a line where the seam is one stabilizer at every $d$, so the curves do not separate downward. At $d=7$ the error runs $0.216$, $0.207$, $0.217$, $0.234$ and $0.266$ across that grid. (b) The adversary reading logical parity
with (solid) and without (dashed) the merge records.}
\label{fig:hw_e3}
\end{figure}

\subsection*{The order in $\gamma$, measured at fixed code}

The comparison above varies the code and the circuit together, so it separates two codes but does not measure the
exponent in \cref{eq:dzlaw}. Measuring the exponent requires the code, the round count and the measured operators to
be held fixed while the damping exposure alone is varied. We use an explicit idle of duration $\tau$ on the data qubits in every round. It is placed immediately before the CNOT layers that imprint the data state on the ancillas, so every recorded check reads a state exposed for the full $\tau$. The per-round damping probability is then
\begin{equation}
\gamma(\tau)=1-\exp\!\big[-(t_{\rm round}+\tau)/T_1\big],
\label{eq:gammatau}
\end{equation}
with $t_{\rm round}=\SI{5.47}{\micro\second}$ read from the backend's own instruction durations (two-qubit
\SI{68}{\nano\second}, measurement \SI{2652}{\nano\second}, reset \SI{2684}{\nano\second}) and $T_1$ from the per-qubit calibration recorded at submission. Each repetition arm is pinned to a line selected for how evenly its qubits decay, so a single rate describes the whole line and the table below quotes it: $\SI{328}{\micro\second}$ at $n=5$ and $\SI{370}{\micro\second}$ at $n=3$. The per-qubit rate vector of the Methods is what the unpinned $[[4,1,2]]$ arm requires, since its four qubits span a factor of $3.5$. The circuits are pinned to a chosen path, so the repetition arms route to $12$ and $24$ two-qubit gates with no routing overhead at all. At every exposure, the two logical-label circuits compile to identical depth and gate count.

\Cref{prop:attained} gives this family's damping-only transcript law in closed form, \cref{eq:repclosed}, with no free parameter to fit, and its $\gamma\to0$ limit $nT\gamma$ carries both the $d_Z=1$ exponent and its coefficient. The
exponent is one number out of that curve. The curve is the whole prediction, and it is the stronger thing to test.

\begin{center}\small
\begin{tabular}{lccccc}
\toprule
$\tau$ (\si{\micro\second}) & $0$ & $3$ & $8$ & $18$ & $35$\\
\midrule
\multicolumn{6}{l}{\emph{repetition $n=5$, $T=1$, $T_1=\SI{328}{\micro\second}$, $40$--$150$k shots per class}}\\
$\gamma$ & $0.0166$ & $0.0255$ & $0.0403$ & $0.0691$ & $0.1162$\\
selection-free estimate & $0.035$ & $0.084$ & $0.175$ & $0.284$ & $0.431$\\
$95\%$ one-sided lower bound & $0.030$ & $0.078$ & $0.165$ & $0.271$ & $0.417$\\
\cref{eq:repclosed} & $0.080$ & $0.121$ & $0.186$ & $0.301$ & $0.461$\\
estimate $\div$ closed form & $0.43$ & $0.69$ & $0.94$ & $0.94$ & $0.94$\\
\bottomrule
\end{tabular}
\end{center}

Over the three exposures where the estimator recovers more than $80\%$ of the curve, the measured log-log slope is $0.85\pm0.03$ against the $0.86$ that \cref{eq:repclosed} itself predicts over that range of $\gamma$. Both slopes come from the same unweighted least-squares fit of $\log\TV$ on $\log\gamma$ across those exposures, the measured one through the point estimates and the predicted one through the closed form at the same $\gamma$. The quoted uncertainty is the standard error of that slope from its residuals. The $n=3$ arm gives $1.07\pm0.07$ against $0.93$, but neither number is the integer $d_Z=1$, as expected because that integer is the $\gamma\to0$ limit while this device's smallest reachable $\gamma$ is $0.015$, where the curve has already left the asymptotic regime. The curve is testable here, and the $n=5$ arm reproduces it at a constant $0.94$ of its height across the upper threefold range of $\gamma$. The ratio falls below that level only where the signal itself approaches the finite-sample floor. The exponent measures the shape of the curve. A constant ratio is informative because the shape is recovered independently of the overall efficiency with which the estimator converts a distribution difference into a threshold statistic.

\paragraph{At $d_Z=2$ the adversary's model class stops being a free choice.}
The same sweep run on a $[[4,1,2]]$ code carries a second-order leak, and it is not visible in the same way. For a
repetition code the two logical classes differ in the \emph{mean} detector density, which any linear score function
resolves. For the $[[4,1,2]]$ code, the classes $(\lvert0000\rangle+\lvert1111\rangle)/\sqrt2$ and $(\lvert1010\rangle+\lvert0101\rangle)/\sqrt2$ have equal mean excitation but differ in the \emph{correlations} between detectors. A
logistic regression recovers $19$--$22\%$ of what a boosted tree recovers from the same shots, reading $0.02$--$0.10$
where the exact law gives $0.07$--$0.47$. On hardware the model family selected on the validation split was the linear
one in one of ten $d_Z=1$ cells and in none of nine $d_Z=2$ cells. An adversary bound reported against a linear score function is therefore not conservative at higher $d_Z$ because it is blind, so the model class has to be selected instead of fixed. This is why the estimator here chooses among four families on a split disjoint from the
one that scores it.

\subsection*{The converse, and the cost of privacy}

Merging two repetition patches by lattice surgery puts $\Zbar_A\Zbar_B$ into the measured algebra
(\cref{fig:hw_e3}). Handing the
adversary the merge records in addition to the split-phase records raises the TV lower bound on the logical parity
from
$0.064$ to $0.924$ at $d=3$, and from $0.312$ to $0.878$ at $d=7$. Two qualifications belong with these numbers. First, on a linear geometry the seam is one stabilizer at every distance, so there is no spatial protection to demonstrate, and the merged-parity error indeed worsens with distance ($p_{LS}=0.100$ at $d=3$, $0.217$ at $d=7$). This is the regime of \Cref{cor:rawreports}, not of the exponential converse. Second, the ``honest only'' arm is itself above the floor because the patches are repetition codes and $d_Z=1$ codes leak. On this device, the converse describes a jump on top of a baseline that is already not private and should be reported that way.

The price of the mitigation is visible in the decoded logical error rate (\cref{fig:hw_cost}), which is a second and
independent exposure
channel. At $n=3$ the plain memory fails at $3.0\times10^{-4}$ for $\zL$ and $3.5\times10^{-2}$ for $\oL$, a factor of $118$, so an observer who sees only \emph{failure rates}, never a syndrome, also learns the logical state. Randomized encoding equalises these to $1.50\times10^{-2}$ and $1.58\times10^{-2}$, within $6\%$, by moving both inputs to their mean, so the cheaper input pays the whole bill.

\subsection*{The theorem's own convergence condition, on this device}

At the most generous admissible model constants, the Koteck\'y--Preiss condition evaluated on the live calibration of the qubit lines actually used gives $\epsilon_0=8.9\times10^{-2}$ against a requirement of $4.1\times10^{-3}$. The equivalent value is $ea=22.4$, whereas $ea<1$ is needed. The two numbers measure different things. The $21.5\times$ is the ratio $8.9\times10^{-2}$ over $4.1\times10^{-3}$, by which the measured $\epsilon_0$ exceeds the largest value the condition admits, and it is the factor every channel rate would have to fall by. The $22.4$ is the polymer activity itself, and it is not $21.5$ because $\etabar$ is nonlinear in $\epsilon_0$. Our sufficient certificate therefore misses on today's hardware, with the depolarizing term dominant, and the diagnostic is not a certification. The condition $ea\ge1$ means the expansion is not known to converge, which is not the same as a leak, and the operational evidence in that regime is the measurement above. Two properties of the number are worth stating. The asymmetry between $P(1\mid0)$ and $P(0\mid1)$ is itself a logical-state-dependent channel that must be measured directly because the device \texttt{Target} exposes only a symmetric readout error and therefore leaves it \emph{understated} here. The reported $\epsilon_0$ uses the worst qubit on each line, which is the right choice for a bound that must hold everywhere on it.

\section*{Relation to correctability--privacy duality, and scope}
\label{sec:complementarity}

A standard route leads from error correction to privacy, and it applies here, so we state the resulting corollary, then
show that it is strictly weaker than what this paper establishes.

For a Stinespring isometry $V:L\to \mathsf{C}\otimes\mathsf{P}$, the correctability--privacy duality states that $L$ is correctable from $\mathsf{C}$ if and only if it is private with respect to $\mathsf{P}$. The approximate version turns $\epsilon$-correctability into $2\sqrt{\epsilon}$-privacy in diamond norm~\cite{kretschmann2008,beny2010,schumacher1996}.

A broadcast classical register consists of two orthogonal, perfectly correlated registers, not one register on both sides of a cut, and both registers occur in the following dilation.
\begin{equation}\label{eq:dilation}
V\lvert\psi\rangle=\sum_{u,\alpha}K_{u,\alpha}\lvert\psi\rangle_{Q_{\rm f}}\otimes\lvert u\rangle_{U_D}
\otimes\lvert u\rangle_{U_O}\otimes\lvert\alpha\rangle_E,
\end{equation}
Here $U_D$ and $U_O$ are the decoder's and observer's respective transcript copies, while $E$ is the remaining environment, and the dilation is an isometry whenever $\sum_{u,\alpha}K_{u,\alpha}^\dagger K_{u,\alpha}=I$, with the bipartition $\mathsf{C}=Q_{\rm f}U_D$ and $\mathsf{P}=U_OE$. Tracing out $\mathsf{P}$ returns the transcript-assisted memory channel $\mathcal{N}_{\mathsf C}(\rho)=\sum_{u,\alpha}K_{u,\alpha}\rho K_{u,\alpha}^\dagger\otimes\lvert
u\rangle\!\langle u\rvert_{U_D}$ because the perfect correlation with $U_O$ decoheres $U_D$ into a classical register. Tracing out $\mathsf{C}$ and then $E$ returns exactly the transcript channel $\mathcal{M}$ of the Methods, since $\sum_\alpha K_{u,\alpha}^\dagger K_{u,\alpha}=E_u$.

A below-threshold memory is correctable from $\mathsf{C}$ because that is what the fault-tolerance threshold theorem supplies, and it is a statement about the recovered logical \emph{qubit} $L_{\rm out}$, not about the classical readout $Y$. Hence:

\begin{proposition}[generic transcript privacy is a corollary]\label{prop:corollary}
If the transcript-assisted memory is $\epsilon_d$-correctable in diamond norm, then
$\|\mathcal{M}-\mathcal{C}\|_\diamond\le2\sqrt{\epsilon_d}$, and in particular a surface-code memory below threshold, with
$\epsilon_d=e^{-\Theta(d)}$, has transcript privacy $e^{-\Theta(d)}$.
\end{proposition}

The existing correctability result in \Cref{prop:corollary} and~\cite{kretschmann2008,beny2010} supplies exponential privacy for ordinary telemetry of a correctable memory.

\subsection*{Where the corollary does not reach}

Correctability is sufficient for privacy, never necessary. Privacy is \emph{anisotropic}, as \Cref{thm:main} shows by assigning each logical axis its own distance, and a code can therefore hide one axis far better than another, whereas correctability sees only $d_{\min}=\min(d_X,d_Y,d_Z)$. Two notions must therefore be kept apart, and we name them here for use throughout.

\begin{description}[leftmargin=1.5em,itemsep=1pt,font=\normalfont\itshape]
\item[Full logical-channel privacy.] $\|\mathcal M-\mathcal C\|_\diamond$ covers \emph{arbitrary} logical inputs including superpositions, mixtures and reference-entangled states, and is controlled by $d_{\min}$, which is the ordinary code distance for a CSS code, while \Cref{thm:main} bounds this quantity and \Cref{prop:corollary} also delivers it.
\item[Basis-label privacy under amplitude damping.] $\TV(P_{0_L},P_{1_L})$ for the two $\Zbar$-\emph{eigenstates},
under the specified damping transcript model. The first permitted order of this quantity is set by $d_Z$ alone (\Cref{cor:basislabel}), a distance that can exceed $d_{\min}$ by an arbitrary amount, and this order is what \cref{eq:dzlaw} computes and what the hardware measures.
\end{description}

The second guarantee protects one logical axis, not the channel, and it is therefore strictly weaker and protects exactly the axis that a $Z$-basis memory exposes to amplitude damping. Two families make the separation sharp. Both have $d_{\min}=1$, so \Cref{prop:corollary} is vacuous for both and neither has any full-channel guarantee at all.

\begin{table}[htbp]
\centering\small
\begin{tabular}{@{}l cccc >{\raggedright\arraybackslash}p{0.34\textwidth}@{}}
\toprule
family & $d_X$ & $d_Y$ & $d_Z$ & $d_{\min}$ & basis-label leak\\
\midrule
repetition, $n$ qubits & $n$ & $n$ & $1$ & $1$ & $nT\gamma+O(\gamma^2)$, first order, attained\\
\addlinespace[2pt]
phase-flip (dual repetition), $n$ qubits & $1$ & $n$ & $n$ & $1$ & \textbf{exactly zero} for the syndrome record, and $\gamma^{\,n}$ once jumps are exposed\\
\bottomrule
\end{tabular}
\caption{The two families that separate the code distance from the $Z$-distance. Both have $d_{\min}=1$, so
correctability supplies no full-channel guarantee for either, yet their basis-label leaks sit at opposite
extremes. Every distance is an exact enumeration of the corresponding logical coset
(\texttt{sim/hw/coset\_distances.py}).}
\label{tab:families}
\end{table}

\Cref{tab:families} collects them, and every entry is an exact enumeration of the corresponding coset rather than a construction-based assertion. Supplementary Note~11.11 tabulates all three distances for these families and for the rotated surface code, using the same code objects that the simulations use. The surface code is itself mildly anisotropic ($d_Y=2d-1$), which is harmless here
but confirms that the componentwise statement is not an artefact of the two degenerate families.

The repetition code stores its logical basis in the damping basis, where a single-qubit population reveals the label, and on hardware the honest transcript recovers that label with a TV lower bound of $0.927$ at $n=21$ under randomised acquisition.

The phase-flip code is its mirror image, and its codewords must be stated carefully because the natural guess is wrong, while its stabilizers are $X_iX_{i+1}$ and its $Z$-stabilizer group is \emph{trivial}, so $\Zbar=Z_1\cdots Z_n$ and $\Xbar=X_1$. The basis-label statement concerns the $\Zbar$-eigenstates. These eigenstates are therefore the cat states
\begin{equation}
\lvert0_L\rangle,\lvert1_L\rangle=\tfrac{1}{\sqrt2}\bigl(\lvert+\rangle^{\otimes n}\pm\lvert-\rangle^{\otimes
n}\bigr),
\end{equation}
and \emph{not} the product states $\lvert\pm\rangle^{\otimes n}$, which are the eigenstates of $\Xbar=X_1$ and agree on \emph{every} $Z$-moment. On the cat pair every moment $\langle Z_S\rangle$ with $|S|<n$ agrees and only
$\langle Z_1\cdots Z_n\rangle=\pm1$ separates them, so no damping-induced record distinguishes them below order
$\gamma^{\,n}$. The coset $\Zbar\cdot\Sgroup$ is the single element $Z_1\cdots Z_n$, giving $d_Z=n$ with multiplicity $1$. The trivial $Z$-stabilizer group makes this exact, and direct enumeration of all $\langle Z_S\rangle$ confirms that the first differing weight is exactly $n$ for $n=3,4,5$.

In fact the separation is sharper than $\gamma^{\,n}$, and in the direction that helps. Every syndrome outcome probability is \emph{identically} independent of the logical input, so the \emph{syndrome} transcript carries exactly zero information for every $\gamma$ and every number of rounds, provided the damping jumps themselves are not exposed (\Cref{prop:xonly}). This follows because this code's measured algebra is entirely $X$-type and the adjoint of amplitude damping maps an $X$-type Pauli to a scalar multiple of itself. The syndrome blindness does not arise because the states are close. One damping layer at $\gamma=0.1$ leaves the two damped codewords $1.70$ apart in trace distance at $n=3$, near the maximum of $2$, so a direct measurement separates them almost perfectly. Only the syndrome record is blind, whereas adjoining the environment's jump pattern makes the leak return at exactly $\gamma^{\,n}$ with coefficient $1$ (\Cref{prop:jump}). Which records are exposed, not how good the code is, decides the outcome.

Note what this does \emph{not} say, since $d_X=1$ leaves the phase-flip code without full logical-channel privacy because a record reading $X_1$ would expose an arbitrary logical input immediately, while the exponential statement is about the $\Zbar$ axis only. The point is that a guarantee stated through one distance cannot express this result, while frameworks that \emph{can} speak axis-by-axis still have a hypothesis that fails on this very family (see below).

No sharpening of \Cref{prop:corollary} recovers these cases because an implication with a false hypothesis says nothing, and we verify below that the same objection applies to the subalgebra form of the duality.
What is needed is a criterion reading the code axis by axis, which is what \Cref{thm:main} provides and
\cref{eq:dzlaw} sharpens. The criterion is channel-relative by construction, with the \emph{amplitude-damping} leak governed by $d_Z$ because that channel acts in the $Z$ basis, while a different dominant channel selects a different axis and a different distance. That is the content of the anisotropic statement, not a qualification of it.

\subsection*{Operator-algebra QEC: a neighbouring notion, and why it is not this one}

Operator-algebra
quantum error correction treats general observable algebras in the Heisenberg
picture~\cite{beny2007oaqec,beny2007algebras}, approximate versions exist~\cite{beny2010algebras,beny2010}, and a
complementarity between private and correctable \emph{subalgebras} is available~\cite{crann2016}, generalising the
subsystem statements~\cite{kribs2008private,kretschmann2008}. In that language, the computational-basis label is represented by the commutative algebra $\mathcal{A}=\mathrm{span}\{I_L,\Zbar\}$. The relevant comparison is with private-algebra privacy for $\mathcal{A}$, whose structural definition differs from the state-distinguishability notion used here.

\emph{A private algebra need not keep its classical centre secret.} Privacy of $\mathcal{A}$ in the operator-algebra sense requires the channel's Heisenberg range, the set of observables it can produce when run backwards on the output, to land in the commutant $\mathcal{A}'$, thereby hiding the \emph{noncommutative} information carried by $\mathcal{A}$ while making no promise about the classical information in the centre $Z(\mathcal{A})=\mathcal{A}\cap\mathcal{A}'$. In plain terms, the algebraic notion hides which superposition was stored and says nothing about which of two basis labels was stored. Here $\mathcal{A}$ is commutative, so it is its own commutant and its centre is all of it, and the notion promises nothing at all about the label. For one logical qubit, $\mathcal{A}=\mathrm{span}\{I_L,\Zbar\}$ is maximal abelian. Thus, $\mathcal{A}'=\mathcal{A}$ and $Z(\mathcal{A})=\mathcal{A}$, so the centre is the entire algebra and the basis label sits precisely where the notion promises nothing. The point is settled by a two-line example in which complete $Z$-dephasing publishes the $Z$-basis populations, while its Heisenberg range lies in the diagonal algebra, which is $\mathcal{A}'$, so $\mathcal{A}$ is private by the standard definition. Yet the two $\Zbar$-eigenstates are mapped to perfectly distinguishable outputs, $\TV=1$ (\texttt{sim/hw/oaqec\_check.py}), and the two notions therefore do not coincide for the commutative logical algebra considered here. In particular, private-algebra privacy does not protect its classical centre.

\emph{What this paper proves is a state-distinguishability statement.} The state-distinguishability quantity is an operational bound on what any test of the transcript can achieve. Our object is $\TV\bigl(P_{U\mid0_L},P_{U\mid1_L}\bigr)$ together with an explicit rate $\kappa_{d_Z}=C\,\mathrm{poly}(d,T)\,\theta^{d_Z/\ell_\ast}$ resolved by the code's coset geometry and the detector graph. The proof derives both the quantity and the rate directly from code coset geometry and the detector graph.

\emph{The algebraic computation provides a comparison and nothing more.} For completeness, we evaluate the Bény--Kempf--Kribs condition for $\mathcal{A}$ to be correctable under the \emph{physical} amplitude-damping channel. The condition is $\code K_a^\dagger K_b\code\in\mathcal{A}'$ for all $a,b$, and maximal abelianness of $\mathcal{A}$ requires every $\code K_a^\dagger K_b\code$ to have vanishing $\Xbar$ and $\bar Y$ components. A Kraus label records which physical qubits jumped, so a pattern of $j$ simultaneous jumps carries amplitude $\gamma^{\,j/2}$. The block $\code K_a^\dagger K_b\code$ acquires an $\Xbar$ component only when one of the two patterns covers a representative of the coset $\Xbar\cdot\Sgroup$, because a jump split across the two patterns annihilates the code space. The enumeration below therefore runs to $d_X$ simultaneous jumps, which is where the leading term lives (\texttt{sim/hw/oaqec\_check.py}).

\begin{center}\small
\begin{tabular}{lcccc}
\toprule
code & $d_Z$ & $d_X$ & largest $|\Xbar|,|\bar Y|$ at $\gamma=10^{-2}$ & order\\
\midrule
phase-flip, $n=3$ & $3$ & $1$ & $4.95\times10^{-2}$ & $\gamma^{1/2}$\\
phase-flip, $n=4$ & $4$ & $1$ & $4.93\times10^{-2}$ & $\gamma^{1/2}$\\
repetition, $n=3$ & $1$ & $3$ & $5.00\times10^{-4}$ & $\gamma^{3/2}$\\
repetition, $n=5$ & $1$ & $5$ & $5.00\times10^{-6}$ & $\gamma^{5/2}$\\
rotated surface, $d=3$ & $3$ & $3$ & $2.43\times10^{-4}$ & $\gamma^{3/2}$\\
\bottomrule
\end{tabular}
\end{center}

No family here is exactly correctable, and the order of the violation is $\gamma^{\,d_X/2}$ in every row, with coefficient $1/2$ for the repetition family and $1/4$ for the surface code. The phase-flip family sits at $d_X=1$, so a single-qubit jump $K_1=\lvert0\rangle\!\langle1\rvert$ already carries $\Xbar$ weight. That is the substance of the comparison, because the algebraic question is governed by $d_X$ while the transcript question this paper answers is governed by $d_Z$, and the table's own columns show the two running in opposite directions. The comparison evaluates whether $\Zbar$ is recoverable on the \emph{output} side of the physical channel, a question defined across a different cut from transcript privacy. The physical channel's environment is closer to the full jump record than to the syndrome record, and \Cref{prop:xonly,prop:jump} show that the two cuts give \emph{different answers on this very family}. The syndrome record is exactly blind, while the syndrome-plus-jump record leaks at $\gamma^{\,d_Z}$. A statement about one does not transfer to the other.

\emph{The resulting claims.} The first is an \emph{operational} pairwise and channel distinguishability bound distinct from algebraic containment. The second is an \emph{axis-resolved} exponent in $d_X,d_Y,d_Z$. The third is an \emph{explicit} detector-graph and Koteck\'y--Preiss rate distinct from an abstract distance. The fourth is a \emph{sufficient criterion for the transcript that invokes no correctability of any axis}. This criterion applies to the phase-flip family and other families without an axis-correctability hypothesis.

\subsection*{Contributions of the axis-resolved framework}

The contribution is the anisotropy and its consequences because \Cref{thm:main} charges each logical Pauli component the distance of \emph{its own} coset instead of charging all three components the code distance. A guarantee stated through $d_{\min}$ alone cannot express this distinction, and the subalgebra route does not reach it on the families that exhibit it. \Cref{prop:dzlaw,prop:attained,prop:xonly,prop:jump} determine the order of the damping leak and the conditions under which it is attained, and they also identify the exposed records that decide between the two outcomes. The regime $d_Z\gg d_{\min}$ is not a pathological corner because biased-noise architectures actually build it, and \Cref{lem:converse} supplies the direction missing from the one-sided duality and names which records expose the input and at what rate. The hardware section measures the leak on a real device and, with matched simulation, it reproduces the state-dependent signal from damping alone, locates the exposure in the syndrome cycle, and prices the cost of removing it. The hypothesis $\GB\subseteq\Sgroup$ is decidable in polynomial time, and the locality indicator is measurable, but a correctability statement supplies neither property.

For Pauli noise the syndrome distribution is \emph{exactly} independent of the logical input by the stabilizer formalism. The contribution treats coherent faults, temporally correlated faults of the bounded-persistence class that PWTS resums, and arbitrary logical inputs, whose fault-path amplitudes require controlling all four code-space matrix elements through \Cref{lem:cosetweight}.

\subsection*{Which secret the theorem covers}

A separate line of work shows that syndrome data can expose \emph{which computation is running}. Shukla, Browne and Nishio~\cite{shukla2026decoder} identify ``gate fingerprints'' in the syndrome stream of Clifford$+T$ surface-code computations, and reconstructing patch activity or dataflow is a further distinct attack surface. These results concern a different secret from the one \Cref{thm:main} covers, and the distinction determines an
operator's safe publication scope. \Cref{tab:scope} sets the two side by side.

\begin{table}[htbp]
\centering
\small
\begin{tabular}{lc}
\toprule
Secret & Covered by \Cref{thm:main} \\
\midrule
Logical state, fixed input-independent schedule & yes, under H1$^\ast$ + PWTS + smallness \\
Identity of the logical gate or circuit being executed & \textbf{no}~\cite{shukla2026decoder} \\
Patch activity, dataflow, resource-usage pattern & \textbf{no} \\
Adaptively or maliciously chosen schedule & \textbf{no} (excluded by hypothesis) \\
Final logical output $Y$ & \textbf{no} (delivered to the user, outside $U$) \\
Direct physical or analog probing of data qubits & \textbf{no} \\
\bottomrule
\end{tabular}
\caption{Scope. The theorem protects the logical input against an observer of the execution transcript on a fixed
schedule. The theorem does not establish telemetry safety because a transcript can hide the logical state and still reveal the algorithm.}
\label{tab:scope}
\end{table}

The practical reading is narrower than ``syndromes can be released''. On a fixed schedule, a code with large
$d_{\min}$ protects an arbitrary logical input, and a code with large $d_Z$ protects the computational-basis label
against amplitude damping specifically. The schedule itself, and everything inferable from it, is not protected at
all.

\section*{Discussion}
The execution transcript of a fault-tolerant computer is a security interface in its own right, and this work marks out two sufficient regimes for it. Syndromes and anything that is a fixed function of them are what we call correctable records, meaning records the decoder reads in order to correct without their carrying the logical state. Under the hypotheses above, the same code distance that corrects errors also hides them, because a contribution that depends on the logical input has to carry a charged operator across a distance-sized region, and the smallness condition makes such regions exponentially rare.
Records calibrated to a logical operator are not. A lattice-surgery parity readout is an intentional logical
measurement that happens to be logged, and it leaks accordingly. Which operators a static patch can place in its measured algebra is itself constrained, and those constraints are what decide where the boundary falls~\cite{shen2026nogo}, so the boundary runs through everyday fault-tolerant primitives instead of exotic ones. That is the relocation principle, and it is the structural content
of this work.

What the hardware adds is a correction to how such a result should be read. Missing the Koteck\'y--Preiss smallness requirement by $21.5\times$ on the measured calibration is a failure of our sufficient certificate, not a demonstration that the physical activity violates the condition. The theorem is conditional, and this sufficient certificate is not satisfied on a present-day superconducting processor. Under randomised, label-balanced acquisition, the honest transcript of a $d_Z=1$ memory identifies its logical input with a TV lower bound of $0.927$ at repetition length $n=21$. The guarantee is thus a design target, and our certificate for it is not met on the processor measured here. Three consequences are practical.

\emph{The code choice is the privacy control.} The basis-label leak order is the code's $Z$-distance. Under the evaluation rule established in the hardware section, the empirical TV lower bound for the $d_Z=3$ circuit on one chip is $11.2$ times smaller than the bound for the $d_Z=1$ circuit. A repetition code stores its logical basis in the noise-preferred basis and is the degenerate worst case, whereas a surface code is not, so the code choice is what sets the leak order, making privacy cheap to obtain and easy to lose.

\emph{The exposure appears to be set by the syndrome cycle instead of the gate set.} On the measured device a round lasts \SI{5.6}{\micro\second}, of which the data qubits spend \SI{5.3}{\micro\second} idle while ancillas are measured and reset, and this idle duration is $40\times$ the gate exposure. The two-qubit layers carry about $2.5\%$ of the round's damping exposure and the idle windows carry the rest, so matched simulation attributes the observed asymmetry to damping during that window. The exposure sweep varies that window directly, and the hardware data are consistent with damping being the dominant mechanism. 
\emph{A cheap mitigation exists, and it is not free.} The opposite codeword is physically prepared under a secret bit and relabelled afterwards. Randomized encoding costs a transversal single-qubit layer and no two-qubit gates ($+5$ one-qubit gates, $+1$ depth at $d=5$) and returns the empirical bound to the statistical floor on good qubit lines when measured on the basis-state pair. On degraded lines, a residual survives, so the claim is reduction by more than an order of magnitude, not removal, while the mitigation's price appears in the logical error rate. The inputs are equalised by moving both to their mean, so whichever was cheaper pays. Two deployment specifications apply. The demonstrated randomization protects the computational-basis label. Protection of an arbitrary logical input uses the full logical Pauli twirl applied physically with a secret key tracked downstream. Relabelling in software alone achieves nothing because a bookkeeping convention leaves the hardware state, and hence the transcript, unchanged.

The guarantee concerns the logical \emph{state} under a fixed, input-independent schedule. Under the same hypotheses, a large $d_{\min}$ protects an arbitrary logical state, while a large $d_Z$ protects the computational-basis label against amplitude damping. Circuit identity, patch activity, adaptive schedules, the final output and direct physical probing are separate observables. Syndrome data can carry gate fingerprints revealing the circuit~\cite{shukla2026decoder}. The schedule and everything inferable from it remain observable.

The threat model is deliberately strong because it gives the adversary the entire transcript and full knowledge of the device, and it also places no bound on the adversary's computation. A weaker adversary seeing only a sampled or coarse-grained log is covered automatically by
the data-processing inequality. The model excludes an adversary able to place a probe on a data qubit, because placing such a probe is physical access instead of a transcript channel. It says nothing about the computational output, which the user must protect by other means.

Relation to prior work. The generic statement is that a correctable memory's transcript is exponentially private. The correctability--privacy duality of approximate quantum error correction~\cite{kretschmann2008,beny2010,schumacher1996} is the closest formal neighbour, and \Cref{prop:corollary} records this existing result. What does not follow from that duality is the axis-by-axis anisotropy of privacy, in which the basis label is governed by $d_Z$ instead of by the code distance, and the two can be arbitrarily far apart. Output-level quantum differential privacy compares two output states under a divergence~\cite{hirche2023,quantumblackwell2025}, while a line of work turns a device's own noise into a privacy budget through inherent gate and readout noise~\cite{zhong2024controllable,ju2024harnessing} and through projection~\cite{li2024projection}. The same line of work uses the error-correction layer itself~\cite{zhong2024,zhong2024dpdqc}, and the connection between the two fields is surveyed in~\cite{zhao2024bridging}, but every one of those budgets is spent at the output. Our object is the internal multi-round transcript, and our mechanism is the homological cost of logical charge instead of an added noise budget. Our bound implies the usual approximate-privacy statements, including $(0,e^{-\Theta(d)})$-differential privacy, because total variation controls every binary test. The operative contribution is the transcript distinguishability bound, with the usual approximate-privacy statements recorded as corollaries.

For delegated and blind quantum computation, the theorem addresses whether the error-correction layer beneath logical encryption exposes the input. Under the stated hypotheses the syndrome telemetry of an honest server is exponentially uninformative about the logical input, while a server exposing a calibrated logical record leaks. Circuit blindness is supplied by the cryptographic protocol~\cite{shukla2026decoder}.

Extending the theorem to other topological and quantum low-density par\-i\-ty check codes~\cite{bravyi1998,bombin2006,bombin2007} should preserve the charge-costs-distance mechanism with a code-dependent constant, while on the empirical side, the measurement that closes the damping check for the repetition family is the fixed-code exposure sweep reported in the hardware section. The sweep, which holds the code fixed while varying the per-round idle duration, uses randomly interleaved and label-balanced acquisition to fit $\mathrm{d}\log\TV/\mathrm{d}\log\gamma$, and it has now been run with an $n=5$ arm (\texttt{sim/hw/submit\_e5b.py}). That arm reproduces the parameter-free curve of \cref{eq:repclosed} at a constant $0.94$ of its height, with a finite-range exponent of $0.85\pm0.03$ against the $0.86$ that curve itself predicts over the reachable range of $\gamma$. Exact computation explains why those are finite-range values and not
integers, and the squeeze is worth recording because it bounds what any such experiment can deliver on present
hardware. At small $\gamma$, the local slope does converge to $d_Z$, but the total variation falls below the finite-sample floor for $d_Z\ge2$, while raising the exposure until it clears the floor pushes $\gamma$ out of the asymptotic regime. Outside that regime, the exact local slope for the $d_Z=3$ patch takes the successive values $2.18$, $1.73$, $0.98$, and $-0.08$ as $\gamma$ goes from $0.28$ to $0.62$ and the total variation saturates. A $d_Z=1$ slope is measurable in tens of seconds of processor time, and a $d_Z=1$ versus $d_Z=2$ contrast is reachable at a few hundred seconds. A $d_Z=3$ slope is not reachable on this device at this shot budget. 
As fault-tolerant machines move from the laboratory to shared infrastructure, their control telemetry becomes a security surface, so the useful question is which parts to expose on which codes and with what measured margin, not whether to expose telemetry. On one real processor, this work gives the size of the gap that remains, and its conditional answer for the surface-code memory supplies a criterion for the code, together with a mitigation whose cost is priced on the repetition-code pair measured here.

\section*{Methods}

\subsection*{The transcript instrument and the privacy metric}
We model a distance-$d$ rotated surface-code memory run for $T=\Theta(d)$ rounds on a fixed, input-independent schedule, and the logical input is one of the two basis states $\zL$, $\oL$. During execution the device emits a classical
\emph{transcript}
\begin{equation}
U=(S,A,R,\tau,J),
\end{equation}
collecting the syndrome history $S$, the decoder actions $A=f(S)$, the ancilla reset records $R$, the timing metadata
$\tau$, and any jump or herald records $J$.

Three registers must be kept apart, and conflating them is the source of a genuine ambiguity that we remove here.
\begin{itemize}[leftmargin=1.5em,itemsep=1pt]
\item $Q_{\rm f}$, the \emph{final physical data register}, consists of the $n=\Theta(d^2)$ data qubits that remain uncorrected at the end of the run.
\item $L_{\rm out}$, the \emph{recovered logical qubit}, is the quantum system obtained by applying to $Q_{\rm f}$ the Pauli-frame recovery $\mathcal{R}_u$ selected by the transcript. This is the object a memory is supposed to preserve,
and correctness of the memory is the diamond-norm statement
$\bigl\|\,\mathbb{E}_u\,\mathcal{R}_u\circ\mathcal{N}_u-\mathrm{id}_L\,\bigr\|_\diamond\le e^{-\Theta(d)}$, which is
what a fault-tolerance threshold theorem supplies. It is a statement about an arbitrary logical input, including
superpositions, not about a bit.
\item $Y$, the \emph{reported logical value}, is the classical outcome of measuring $L_{\rm out}$ in the logical $Z$ basis, produced only when the user asks for a computational-basis readout. $Y$ is a readout \emph{of} $L_{\rm out}$,
not $L_{\rm out}$ itself.
\end{itemize}
The transcript excludes all three, which means that $U$ contains neither $Q_{\rm f}$, nor $L_{\rm out}$, nor $Y$, and recovery is performed by the user or by a trusted agent holding $Q_{\rm f}$. In contrast, the observer considered in this paper sees only $U$.

The distinction is load-bearing twice over. The distinction makes the privacy claim a diamond-norm statement about the channel $L\to U$ instead of two probability distributions, so superpositions and mixtures are covered at the same rate. It also makes the correctable side of the complementarity argument (Relation section) a quantum subsystem instead of a classical bit, without which that argument could not even be posed.

We describe the transcript using detector events $\delta_t=s_t\oplus s_{t-1}$ of the three-dimensional decoding graph, and in these coordinates a persistent data error appears only when it starts and ends. It is not reported every round.

The map $\mathcal{M}:\rho\mapsto U$ is a quantum instrument, a completely positive map whose classical output is the
transcript and whose quantum output is discarded, followed by fixed classical post-processing. Writing $\{E_u\}$ for its effects, the probability of transcript $u$ on input $\rho$ is $\Tr(E_u\rho)$. The decoder action $A=f(S)$ and any timing or reset metadata are deterministic or input-independent functions applied on top, so the data-processing inequality prevents them from increasing distinguishability. The running example is the basis pair $\zL,\oL$. However, the bound itself is not limited to that pair (\Cref{lem:offdiag}). For that pair, privacy is the distinguishability of the two transcript laws
$\Pzero$ and $\Pone$ under total variation,
\begin{equation}
\TV(\Pzero,\Pone)=\tfrac12\sum_u\bigl|\Pzero(u)-\Pone(u)\bigr|=\tfrac12\sum_u\bigl|\Tr(E_u\Delta)\bigr|,
\qquad \Delta=\rho_0-\rho_1,
\end{equation}
where $\{E_u\}$ are the instrument effects. For a binary input, total variation has a direct testing meaning because the Neyman--Pearson lemma implies that every test deciding $\zL$ versus $\oL$ from $U$ has type-I plus type-II error at least $1-\TV$. The best distinguishing advantage is therefore $\TV/2$.

\paragraph{Terminology.}
The record passes through a fixed hierarchy, and each stage has a name used consistently throughout.
\begin{equation*}
\begin{aligned}
&\underbrace{\text{raw measurement outcomes } s_t}_{\text{per-round ancilla readouts}}
\;\longrightarrow\;
\underbrace{\text{detector events } \delta_t=s_t\oplus s_{t-1}}_{\text{the decoding graph's vertices}}\\[6pt]
&\qquad\longrightarrow\;
\underbrace{\text{decoder actions and metadata } A=f(S),\,R,\,\tau,\,J}_{\text{fixed functions of the above}}
\;\longrightarrow\; U .
\end{aligned}
\end{equation*}
The remaining terms are defined once here.

\begin{description}[leftmargin=1.5em,itemsep=1pt,font=\normalfont\itshape]
\item[Fault-support region $\Rlat$.] The connected union of the detector \emph{cells} that a bra--ket fault pair touches, taken \emph{before} the mod-2 cancellation that produces the observed syndrome, is the object for which the polymer expansion, PWTS and \Cref{lem:cosetweight} are stated throughout.
\item[Observed component $\Robs$.] A connected component of \emph{fired} detectors in a sampled transcript is what the empirical diagnostic measures, but it is \emph{not} the same object. A length-$k$ fault chain occupies $\Theta(k)$ cells but fires only its endpoints, whereas a homologically nontrivial loop can occupy $\Theta(d)$ cells while firing nothing at all. Every geometric statement below uses the fault-support region $\Rlat$. By contrast, the observed component $\Robs$ appears only in the empirical diagnostic.
\item[Neutral / charged.] A region is \emph{neutral} when the operator it carries is homologically trivial and \emph{charged} otherwise, while neutrality is a homological property, \emph{not} zero syndrome, so a logical loop with trivial boundary and nontrivial homology is charged.
\item[Backbone.] The stabilizer gadgets, ancilla resets and classical post-processing make up the error-free part of the syndrome-extraction circuit, and the separate fault paths are inserted on top of this backbone.
\item[Persistent wire, or rod.] A fault that survives many rounds at bounded spatial footprint. In detector coordinates, it appears only where it starts and ends, so it is one bounded-size polymer instead of a chain whose degree grows with $T$. Resumming persistent rods makes the expansion converge uniformly in $T$.
\item[Calibrated witness.] A transcript record that tracks a known physical operator with a known, bounded error, so
that an adversary can combine such records without first having to learn the device. It is \emph{charged} when the
product of the operators it tracks is a nontrivial logical.
\item[PWTS.] \emph{Persistent-wire transcript summability} is the hypothesis that the resummed activity $\eta(R)$ of the fault histories producing any connected fault-support region $\Rlat$ decays geometrically in $|\Rlat|$, uniformly in $d$ and $T$ (\Cref{hyp:pwts}).
\end{description}

\subsection*{The two transcript-locality hypotheses}
The privacy statement is conditional on two explicit hypotheses on the observed instrument, both of which are properties of the honest execution, stated once and carried through the proof.

The first hypothesis, H1$^\ast$, says the honest part of the device is blind to the logical input. In plain terms, the records produced by error-free syndrome extraction look the same for every stored logical state because the measured operators are scalar on the code space.

\begin{assumption}[H1$^\ast$: sector-scalar honest backbone]\label{hyp:h1}
For every observed outcome $u$, the honest-backbone effect is a stabilizer polynomial that is scalar within each syndrome sector,
\begin{equation}
F^B_u=\sum_{\sigma} q_u(\sigma)\,\Pi_\sigma \in \mathbb{C}[\Sgroup],
\end{equation}
where $\Pi_\sigma$ projects onto the syndrome sector $\sigma$ and $\Sgroup$ is the stabilizer group. Equivalently, the sectorwise Knill--Laflamme condition $\code\, E_\sigma^\dagger F^B_u E'_\sigma\,\code=\lambda_{\sigma,u,E,E'}\,\code$ holds, and ``neutral'' means trivial homology, with zero syndrome alone insufficient. H1$^\ast$ holds for circuits with complete
stabilizer gadgets, fresh or reset ancillas, classical-kernel records, decoder post-processing $A=f(S)$, and no
intermediate physical logical recovery.
\end{assumption}

Two features of the condition matter for what follows, the first being its sectorwise form, which is necessary because a summed Knill--Laflamme condition asks only that the backbone be scalar after averaging over syndromes. The coefficients $q_u(\sigma)$ vary from sector to sector, so a sector-dependent phase could carry logical information hidden by the average. The condition we use pins the backbone to a scalar within each syndrome sector separately. The
second is that ``neutral'' is a homological notion. A fault that produces zero net syndrome can still be a logical loop,
which has trivial boundary but nontrivial homology, and such a loop is not neutral. Reading neutrality as zero syndrome
instead of trivial homology would wrongly classify a logical operator as harmless, so the distinction is
load-bearing.

One part of H1$^\ast$ can be checked by binary-symplectic algebra. The measured-check group $\GB$ must lie inside the
stabilizer group $\Sgroup$, so that no product of measured checks is a nontrivial logical. This is a polynomial-time
linear-algebra test over the binary field, and it is exactly the first half of the diagnostic. The full hypothesis is an explicit assumption covering coherent faults, analog records, leakage, timing, and feedback, so the $\GB\subseteq\Sgroup$ test supports but does not establish it. Those richer records are not products of measured checks and therefore lie outside the reach of the algebra test. This gap between the checkable part and the full hypothesis makes a second, calibrated part necessary for the diagnostic. It also explains why we present H1$^\ast$ as a hypothesis and not as a proved property of generic hardware.

The second hypothesis, PWTS, says faults stay local in the transcript because a region of the record that could carry logical information has an activity that shrinks geometrically with the region's size.

\begin{assumption}[PWTS: persistent-wire transcript summability]\label{hyp:pwts}
There are constants $K_{\rm rod}$ and $\etabar<1$, both independent of $d$ and $T$, such that for every connected
fault-support region $\Rlat$ (correctable or not) the pre-logical rod activity obeys
\begin{equation}
\eta(R)\le (K_{\rm rod}\,\etabar)^{|R|}.
\end{equation}
Here $\eta(R)$ is the resummed absolute norm of the non-blind fault-path contributions whose detecting support is $R$, taken \emph{before} the homological-charge test. In plain terms, it is the total weight of all local fault histories that produce the same connected detector region $R$, before asking whether that region is logical. The bound is on this
pre-logical activity, so it applies to charged regions as well. The logical quotient of a correctable region is already zero by H1$^\ast$ and Knill--Laflamme, and the
cluster expansion needs the pre-logical activity of every region.
\end{assumption}

PWTS plays the role of the standard locality and noise assumptions of fault-tolerance threshold
theorems~\cite{aliferis2006,terhalburkard2005,aliferis2007}, now for the transcript. PWTS has a model derivation for the standard reset-ancilla syndrome-record model under the stated microscopic locality and rate assumptions. Once the exact transcript factorization is established, PWTS holds with the explicit expression
\begin{equation}
\etabar=e^{c_{\rm blk}\varepsilon_0}-1\le 2c_{\rm blk}\varepsilon_0,\qquad
\varepsilon_0=\varepsilon_{\rm dep}+\varepsilon_{\rm meas}+\varepsilon_{\rm reset}+\varepsilon_{\rm coh}
+\varepsilon_{\rm xtalk}+\varepsilon_{\rm leak}.
\end{equation}
Each $\varepsilon$ is the local unnormalized non-identity branch weight of its noise channel (for coherent over-rotations $\varepsilon_{\rm coh}=\sup_v 2|\sin(\theta_v/2)|\le\sup_v|\theta_v|$). The derivation is given in
Supplementary Note~2. Outside the stated model PWTS must be separately certified.

\subsection*{Privacy notion}
The no-leakage theorem bounds the induced classical channel from the logical qubit to the transcript within $\kappa_d=e^{-\Theta(d)}$ of an input-independent channel in the diamond norm, so $\TV(P_\rho,P_\sigma)\le\kappa_d$ for \emph{every} pair of logical inputs, including pairs beyond $\zL,\oL$. We state the operational consequence in
hypothesis-testing form~\cite{dong2022gdp,kairouz2015,wasserman2010} and record the differential-privacy statements as
corollaries~\cite{dwork2006,hirche2023}.

\begin{theorem}[hypothesis-testing privacy]\label{thm:fdp}
If $\TV(\Pzero,\Pone)\le\kappa_d$, the transcript experiment is $f_{\kappa_d}$-private with trade-off $T(\alpha)\ge(1-\kappa_d-\alpha)_+$. Consequently, every test has type-I plus type-II error at least $1-\kappa_d$, so the optimal distinguishing advantage is at most $\kappa_d/2$. It follows that the experiment is $(0,\kappa_d)$-DP, with
$E_{e^{\varepsilon}}(\Pzero\Vert\Pone),\,E_{e^{\varepsilon}}(\Pone\Vert\Pzero)\le\kappa_d$ for all $\varepsilon\ge0$.
\end{theorem}

The proof is one line of Neyman--Pearson (Supplementary Note~3). The trade-off function is the primitive here, and the DP parameters are summaries of it, so we do not present the DP corollary as a separate contribution. Pure $\varepsilon$-DP and R\'enyi-DP are unavailable for a structural reason instead of a technical one. A faithful model has transcripts occurring under one input and not the other, which forces an infinite likelihood ratio and makes a nonzero additive term unavoidable. Supplementary Note~3 also records a smooth $(\varepsilon_d,\delta_d)$-DP corollary with both parameters $e^{-\Theta(d)}$. It states that off an exponentially rare event the privacy loss is exponentially small, which is weaker than a relative-likelihood statement and is not claimed as one.

Two features of the experiment differ from the classical database setting and should be stated, the first being that the \emph{total} adjacency relation makes every logical input a neighbour of every other, which is stronger than differing-in-one-record adjacency. The diamond-norm form also covers an adversary holding a quantum register entangled with the logical qubit because the norm is taken over an arbitrary reference system (Supplementary Note~11.10). The classical-transcript adversary of the Results is the special case with a trivial reference.

\subsection*{Device-model derivation of the locality hypothesis}
We derive the PWTS bound for the standard reset-ancilla syndrome-record model, with the full derivation in Supplementary Note~2 and a summary here. The model assumption is that the adversary-visible transcript factors through reset-ancilla
detector events and input-independent classical post-processing, with no exposed record calibrated to a data-local
operator. Under this factorization, one blocks a complete syndrome-extraction gadget and bounds the unnormalized outcome-summed weight of its non-identity branches by
\begin{equation}
\etabar=e^{c_{\rm blk}\varepsilon_0}-1\le 2c_{\rm blk}\varepsilon_0,\qquad
\varepsilon_0=\varepsilon_{\rm dep}+\varepsilon_{\rm meas}+\varepsilon_{\rm reset}+\varepsilon_{\rm coh}
+\varepsilon_{\rm xtalk}+\varepsilon_{\rm leak}.
\end{equation}
Each $\varepsilon$ is the local non-identity branch weight of one noise channel, and a coherent over-rotation of angle $\theta$ contributes $\varepsilon_{\rm coh}=\sup_v 2|\sin(\theta_v/2)|\le\sup_v|\theta_v|$. A persistent data error registers only at its detector endpoints and occupies a bounded detector footprint, so its temporal history resums once into a spatial rod. This happens because $\delta_t=s_t\oplus s_{t-1}$ is zero while the syndrome is constant. The polymer-activity parameter $a=e\nu K_{\rm rod}\etabar$ then depends on the physical rates as
$\etabar=O(p)+O(|\theta|)+O(\chi)+O(\lambda)$, with $p$ the stochastic rate, $\theta$ the coherent angle, $\chi$ the
crosstalk strength, and $\lambda$ the leakage rate. The factorization assumption used here excludes the persistent quantum-nondemolition counterexample. In this counterexample, a record reports a data-local factor every round, thereby reconstructing a logical operator over $\Theta(d)$ rounds and driving the one-wire activity to one. Outside the stated
model the bound must be certified separately, which is the role of the hardware calibration in the diagnostic.

The derivation also draws a clean boundary between records that satisfy the hypothesis and records that violate it, and
the boundary is physical, not formal. On the satisfying side are the ordinary ingredients of honest local
execution. Reset-ancilla syndrome detectors satisfy it because each fault has a bounded detector footprint in the three-dimensional decoding graph, while ancilla-measurement records satisfy it because they are stabilizer-check outcomes, sector-scalar under H1$^\ast$. Local reset and herald flags, fixed timing metadata, and the decoder's own actions
satisfy it because they are input-independent classical functions, and classical post-processing cannot increase the
total variation. Circuit-level depolarizing noise, small coherent over-rotations, finite-range crosstalk, and
leakage-reduced leakage satisfy it because each contributes a bounded-footprint atom, enlarging the constants but not
the scaling. On the violating side are records that report a logical factor repeatedly or persistently, including direct data-qubit dispersive readout and per-qubit energy or amplitude telemetry, which drive the one-wire activity toward one and break the bound. A persistent sensor tied to a data qubit and mobile leakage that carries a measurement record likewise drive the one-wire activity toward one and break the bound. A useful borderline case is low-frequency flux noise, which remains harmless when unobserved or exposed only as an input-independent nuisance that affects the transcript through syndrome faults. It is dangerous only if its exposed record becomes correlated with a data-local logical factor. The lesson the boundary teaches is that the hypothesis is
not about how noisy the device is, but about whether its records repeatedly read out the data.

\subsection*{Proof architecture of the no-leakage theorem}
The bound $\TV=\tfrac12\sum_u|\Tr(E_u\Delta)|$ is controlled by four structural steps and a cluster expansion,
reduced to three lemmas proved in Supplementary Note~1. The naive expansion has terms of every order, and a bound controlling only pairwise correlations would miss the dangerous high-order terms in which many individually invisible local records combine into a logical string. By contrast, the cluster expansion controls all orders at once. In this expansion, a polymer denotes a connected cluster of detector events. We work throughout with the absolute outcome-summed norm, which turns the signed sum over fault paths into a positive polymer gas and therefore relies on no cancellation between records.
\begin{enumerate}[leftmargin=1.5em,itemsep=2pt]
\item \emph{Neutral contributions vanish.} By Knill--Laflamme, a bra--ket pair with homologically trivial support has
code-space matrix elements proportional to the identity within its syndrome sector, so it cancels in
$\Tr(E_u\Delta)$ (\Cref{rem:supports}).
\item \emph{The honest backbone is blind (Lemma L1).} Under H1$^\ast$ the backbone effect is sector-scalar. Its
off-code dressing carries no logical signal, so only the non-blind quotient remains.
\item \emph{Charge costs distance.} Homology forces any charged spacetime cluster to have support $|C|\ge d_Z/\ell_\ast$ for a constant $\ell_\ast$, where $d_Z$ is the minimum weight of the coset $\Zbar\Sgroup$. For the rotated surface code, this relation reduces to $d_Z=d$.
\item \emph{No cancellation (Lemma L2).} Working with the absolute outcome-summed norm yields a signed-polymer
expansion with a Koteck\'y--Preiss partition-ratio bound~\cite{koteckypreiss1986}, so the neutral exterior cancels
uniformly in the spacetime volume.
\item \emph{Subcritical activity (Lemma L3).} Counting by connected detecting-region size, a direct
Mayer--Penrose--Fern\'andez--Procacci tree expansion~\cite{penrose1967,fernandezprocacci2007} with persistent errors
resummed into bounded spatial footprints gives, under PWTS, a region-activity sum $\sum_{|R|=n}\eta(R)\le(e\nu)^{-1}
a^n$ with $a=e\nu K_{\rm rod}\etabar$.
\end{enumerate}
Combining the charged-support lower bound $|C|\ge d_Z/\ell_\ast$ with the subcritical activity and the
Koteck\'y--Preiss convergence yields, for $T=\Theta(d)$,
\begin{equation}
\TV(\Pzero,\Pone)\le C\,\mathrm{poly}(d,T)\,\theta^{d_Z/\ell_\ast}=e^{-\Theta(d_Z)},
\end{equation}
which is the no-leakage theorem. Here $\nu$ is the region connective constant, and $K_{\rm rod}$ is the PWTS activity constant. Convergence requires the Koteck\'y--Preiss condition $ea<1$, and all constants are independent of $d$ and $T$ because PWTS bounds region activities uniformly.

A triangle inequality over basis states would not extend this result to arbitrary inputs because the measurement $\{(I\pm\Xbar)/2\}$ has identical laws on $\zL,\oL$ while separating $\lvert\pm_L\rangle$, so the off-diagonal code-space elements must be controlled directly.

\begin{definition}[coset distances]\label{def:cosetdist}
For a stabilizer code with stabilizer group $\Sgroup$ and logical Paulis $\Xbar,\bar Y,\Zbar$, set
\begin{equation}\label{eq:cosetdist}
d_P:=\min\bigl\{\mathrm{wt}(g):g\in\bar P\cdot\Sgroup\bigr\},\qquad \bar P\in\{\Xbar,\bar Y,\Zbar\},
\qquad d_{\min}:=\min(d_X,d_Y,d_Z).
\end{equation}
Every element of $\Norm\setminus\Sgroup$ lies in exactly one of the three cosets, so $d_{\min}$ is the ordinary code distance and $d_P\ge d_{\min}$ for each axis. However, the inequality can be strict by an arbitrary amount (Supplementary Note~11.11).
\end{definition}

The single step on which everything turns is that a nonzero logical-$\bar P$ component costs $d_P$. The controlling quantity is not $d_{\min}$. This step is worth isolating because the obvious argument does not give it and the operators it must cover are not Paulis.

\begin{lemma}[a charged component costs its own coset's distance]\label{lem:cosetweight}
Let $O$ be any operator that is supported on a qubit set $A$ and may in particular be coherent and non-Pauli. If
$\Tr_L\bigl(\bar P\,\code\,O\,\code\bigr)\ne0$ for some $\bar P\in\{\Xbar,\bar Y,\Zbar\}$, then $|A|\ge d_P$, and the
connected fault-support region $\Rlat$ carrying $O$ satisfies $|\Rlat|\ge d_P/\ell_\ast$.
\end{lemma}
\begin{proof}
Expand $O$ in the $n$-qubit Pauli basis, $O=\sum_E c_E E$, and since $O$ is supported on $A$, every $E$ with $c_E\ne0$ has $\mathrm{supp}(E)\subseteq A$. A Pauli's code-space block is determined by its class. If $E\notin\Norm$, then $\code E\code=0$ because it anticommutes with some stabilizer $s$, so $\code E\code=\code sEs\code=-\code E\code$. The relation $\code E\code=\pm\code$ holds if $E\in\Sgroup$. If $E$ lies in one of the three nontrivial cosets, then $\code E\code$ is the corresponding logical Pauli up to a phase. Since $\Tr_L(\bar P\bar Q)=2\delta_{PQ}$, only Paulis in the \emph{single}
coset $\bar P\cdot\Sgroup$ contribute to $\Tr_L(\bar P\,\code\,O\,\code)$. The hypothesis therefore forces some
$E\in\bar P\cdot\Sgroup$ with $c_E\ne0$, and $\mathrm{wt}(E)\ge d_P$ by \cref{def:cosetdist}; as
$\mathrm{supp}(E)\subseteq A$, $|A|\ge d_P$. The region satisfies $|\Rlat|\ge d_P/\ell_\ast$ because each detector cell covers at most $\ell_\ast$ data qubits. By definition, $\Rlat$ is the union of the cells the pair \emph{touches}, so every qubit of $A$ lies in some cell of $\Rlat$ and $|A|\le\ell_\ast|\Rlat|$.
\end{proof}

\begin{remark}[why this must be $\Rlat$ and not the fired detectors]\label{rem:latent}
The inequality $|A|\le\ell_\ast|\Rlat|$ is false for the observed component $\Robs$, and the failure is not marginal because detector values are mod-2 parities, so a chain of faults cancels in its interior. On the repetition code's detector graph, a length-$8$ chain touches $8$ cells and fires a single detector, and a logical representative can occupy its whole support while firing nothing. Bounding the support by the fired set would therefore give no distance cost at
all. The activity is summed over \emph{all} outcome patterns on the touched cells, including the silent pattern. The polymer expansion groups faults into connected unions of touched cells so that a homologically charged but syndrome-silent configuration still pays $|\Rlat|\ge d_P/\ell_\ast$. This is the step that
carries the coset distance into the exponent, and it is the one place where the distinction is mathematically essential.
\end{remark}

\begin{remark}[why the one-line argument is not enough]\label{rem:notkl-main}
Knill--Laflamme gives only the weaker $|A|\ge d_{\min}$. If $|A|<d_{\min}$, then $\code O\code=c\code$ is a scalar and \emph{every} traceless logical component vanishes at once, which discards the axis label and returns the same distance for all three axes. It is also tempting to argue that a nonzero $z_u$ forces $\code K_b^\dagger K_a\code\propto\Zbar$
and hence that $K_b^\dagger K_a$ ``lies in'' the coset $\Zbar\Sgroup$. For Pauli fault paths, that argument is harmless, but a coherent or non-Markovian fault path makes $K_b^\dagger K_a$ a \emph{superposition} of sectors that lies in no single coset. A nonzero component says only that its projection onto that sector is nonzero, and \Cref{lem:cosetweight} covers this case because it is stated for arbitrary operators. Its proof needs only the existence of one Pauli term in the coset and imposes no bound on how many terms occur, so it introduces no dimensional constant.
\end{remark}

\begin{lemma}[componentwise off-diagonal block control]\label{lem:offdiag}
Write the code-space block of each transcript effect as $M_u=q_u I_L+x_u\Xbar+y_u\bar Y+z_u\Zbar$, and let
$\mathcal{M}(\rho)=\sum_u\Tr(M_u\rho)\,\lvert u\rangle\!\langle u\rvert$ be the induced classical channel from the
logical qubit to the transcript, with $\mathcal{C}(\rho)=\sum_u q_u\Tr(\rho)\,\lvert u\rangle\!\langle u\rvert$ the
input-independent channel. Under the hypotheses of \Cref{thm:main}, each logical Pauli component carries the distance cost \emph{of its own coset}, through the bounds
\begin{equation}\label{eq:componentwise}
\sum_u|x_u|\le\kappa_{d_X},\qquad \sum_u|y_u|\le\kappa_{d_Y},\qquad \sum_u|z_u|\le\kappa_{d_Z},
\qquad \kappa_m:=C\,\mathrm{poly}(d,T)\,\theta^{m/\ell_\ast}.
\end{equation} Consequently, the following inequality holds over any reference system and for any normalized input.
\begin{equation}
\|\mathcal{M}-\mathcal{C}\|_\diamond\le\sum_u\|M_u-q_uI_L\|_\infty\le
\kappa_{d_X}+\kappa_{d_Y}+\kappa_{d_Z}\le 3\,\kappa_{d_{\min}}.
\end{equation}
This inequality implies $\TV(P_\rho,P_\sigma)\le 3\kappa_{d_{\min}}$ for all logical inputs $\rho,\sigma$, including superpositions and mixtures (\Cref{cor:fullchannel}). For the two $\Zbar$-eigenstates, only $z_u$ survives, giving $\TV(\Pzero,\Pone)\le\kappa_{d_Z}$ (\Cref{cor:basislabel}).
\end{lemma}

\begin{proof}[Proof of \Cref{lem:offdiag}, hence of \Cref{thm:main}]
Expand the effect over bra and ket fault paths, $E_u=\sum_{a,b}c_{ab}(u)\,K_b^\dagger F^B_u K_a$, where $K_a,K_b$ are
fault-path operators and $F^B_u$ is the honest-backbone effect. Projecting onto the code space and reading off the
logical Pauli components,
\begin{equation}
x_u=\tfrac12\Tr\bigl(\Xbar\,\code E_u\code\bigr),\qquad\text{and similarly for }y_u,z_u .
\end{equation}
Fix one pair $(a,b)$ and expand both fault-path operators in the physical Pauli basis,
$K_a=\sum_{E}\alpha_E E$ and $K_b=\sum_{E'}\beta_{E'}E'$. It is \emph{not} legitimate to argue that $F^B_u$, being a scalar within each syndrome sector, simply drops out and leaves the logical content to $K_b^\dagger K_a$. Writing $F^B_u=\sum_\sigma q_u(\sigma)\code_\sigma$ gives
\begin{equation}\label{eq:sectorsum}
\code\,K_b^\dagger F^B_u K_a\,\code=\sum_\sigma q_u(\sigma)\,\code\,K_b^\dagger\code_\sigma K_a\,\code .
\end{equation}
Contributions that cancel between sectors inside $\code K_b^\dagger K_a\code$ need not cancel once weighted by distinct $q_u(\sigma)$. Indeed $\code K_b^\dagger K_a\code=0$ is compatible with
$\code K_b^\dagger F^B_u K_a\code\propto\Zbar$ (\Cref{rem:sectorcaution}).

What survives is stronger than what that step assumed, and it needs no cancellation at all. Each syndrome
projector is a combination of stabilizers, $\code_\sigma=|\Sgroup|^{-1}\sum_{s\in\Sgroup}\chi_\sigma(s)\,s$, and
for any Paulis $E,E'$ and any $s\in\Sgroup$,
\begin{equation}\label{eq:stabpass}
\code\,E'^\dagger s E\,\code=\pm\,\code\,E'^\dagger E\,\code ,
\end{equation}
since $s$ commutes or anticommutes past $E$ and $s\code=\code$. A syndrome projector can therefore only \emph{rescale} the coefficient of a Pauli pair, and it can never introduce a Pauli term that was not already present.
Collecting \cref{eq:sectorsum,eq:stabpass},
\begin{equation}
\code\,K_b^\dagger F^B_u K_a\,\code=\sum_{E,E'}\overline{\beta_{E'}}\,\alpha_E\;w_u(E,E')\;\code\,E'^\dagger E\,\code ,
\end{equation}
with scalars $w_u(E,E')$ collecting the sector weights and signs. Hence a nonzero logical-$\bar P$ component
forces the \emph{existence} of a Pauli pair with $\Tr_L(\bar P\,\code\,E'^\dagger E\,\code)\ne0$, that is
$E'^\dagger E\in\bar P\cdot\Sgroup$, so $\mathrm{wt}(E'^\dagger E)\ge d_P$ by \cref{def:cosetdist}. Since
$\mathrm{supp}(E'^\dagger E)\subseteq\mathrm{supp}(E)\cup\mathrm{supp}(E')$, that weight is carried inside the
bra--ket \emph{joint} support $A$ of the pair (\Cref{rem:supports}), giving $|A|\ge d_P$ and, by
\Cref{lem:cosetweight}, $|\Rlat|\ge d_P/\ell_\ast$ for the region carrying it. The same with $\Xbar$ and $\bar Y$ gives $d_X/\ell_\ast$ and $d_Y/\ell_\ast$, and each logical axis is charged with respect to its own coset and pays its own distance.

It remains to sum the activities of all such regions, and the bound is taken before the homological test, which makes it available for charged regions. PWTS bounds the pre-logical activity of \emph{every} connected region by $\eta(R)\le(K_{\rm rod}\etabar)^{|R|}$. The Koteck\'y--Preiss condition $ea<1$ yields a convergent cluster expansion in which the total weight of charged clusters of size at least $m/\ell_\ast$ is bounded by $\kappa_m=C\,\mathrm{poly}(d,T)\,\theta^{m/\ell_\ast}$, with $\theta<1$. This expansion sums over connected regions anchored anywhere in the spacetime volume, using connective constant $\nu$ and polymer activity $a=e\nu K_{\rm rod}\etabar$, while the polynomial factor counts anchor positions. Applying this with $m=d_X,d_Y,d_Z$ respectively gives \cref{eq:componentwise}. Finally, $\|M_u-q_uI_L\|_\infty\le|x_u|+|y_u|+|z_u|$, and summing over $u$ gives the stated diamond-norm bound. A channel with classical output is bounded by that quantity because the output registers $\lvert u\rangle\!\langle u\rvert$ are mutually orthogonal, and adjoining a reference system does not increase it. The tail bound itself, its polynomial prefactor and the identification $\theta=a$ are derived in Supplementary Notes~11.5 to~11.7. Detailed constants are in Supplementary Note~11.3 (the coset distances and the charge lemma), 11.9 (the componentwise bound), and 11.10 (the diamond-norm form). Supplementary Note~11.11 tabulates $d_X,d_Y,d_Z$ for every code family used in this work, by exact enumeration.
\end{proof}

\begin{remark}[why the sector weights cannot be dropped]\label{rem:sectorcaution}
The tempting short step is ``$F^B_u$ is sector-scalar, so the logical content comes from $K_b^\dagger K_a$'', but this step is false, as an explicit two-qubit example shows. Take $\Sgroup=\langle Z_1Z_2\rangle$ with $\Zbar=Z_1$,
and $K_a=I+X_1$, $K_b=Z_1-X_1Z_1$. Then $K_b^\dagger K_a=0$ identically, so $\code K_b^\dagger K_a\code=0$. However, with $F^B=q_+\code_++q_-\code_-$ one finds $\code K_b^\dagger F^B K_a\code$ has $\Zbar$-component $\tfrac12(q_+-q_-)$, nonzero whenever the backbone weights the two sectors differently. The proof above avoids this because it uses only the \emph{existence} of a pair in the relevant coset and never invokes cancellation between Pauli pairs. The example is consistent with the conclusion because the pair $Z_1\cdot I$ lies in $\Zbar\cdot\Sgroup=\{Z_1,Z_2\}$, whose minimum weight is $d_Z=1$, and the joint support is a single qubit.
\end{remark}

\begin{remark}[which support is meant]\label{rem:supports}
The phrase ``below the distance'' is used loosely in the literature, and four distinct notions appear above, but the proof given here uses only the fourth of these notions. A distance-$d$ code corrects arbitrary errors of weight at most $\lfloor(d-1)/2\rfloor$, but the
statement that actually enters is local indistinguishability, $\code\,O\,\code=c_O\code$ for $O$ supported on a
homologically trivial region. Concretely, (i) the weight of a \emph{single} fault path is not the relevant quantity, while (ii) Knill--Laflamme uses the joint support of the \emph{bra--ket pair} $K_b^\dagger K_a$, which can be up to twice either path. (iii) The expansion is organised by the \emph{fault-support region} $\Rlat$ occupied by the pair, and PWTS bounds that region, while (iv) \emph{Homological charge}, not support size as such, determines the exponential cost.
A fault of large support that is homologically trivial is free, and a fault of support exactly $d$ that wraps the
torus is not. Reading (iv) as (i) would make the argument look like a claim that ``every contribution below distance
$d$ vanishes'', which is neither what is proved nor true.
\end{remark}

The imaginary part $y_u$ is the coherent off-diagonal contribution, controlled by the same charge cost with no
cancellation assumed, so coherent interference cannot leave a residual $\langle 0_L|E_u|1_L\rangle$ term. The diamond
norm is taken with the standard normalization $\|\rho-\sigma\|_1\le2$, the channel output is the classical transcript
register $\lvert u\rangle\!\langle u\rvert$, and the supremum over reference systems is what certifies robustness against
a quantum-side-information adversary.

Two features of this architecture deserve emphasis, because they are where a naive argument fails. The first is that
the expansion is over connected detecting regions counted by their support size, not over individual faults counted by
number. A persistent error that lasts many rounds is a single rod of bounded spatial footprint, paid once through its
PWTS activity, not a long chain of vertices whose degree grows with the number of rounds. Without this resummation the
polymer degree would grow with $T$ and the expansion would not converge uniformly, which is the precise reason that the
hypothesis PWTS, not a per-fault smallness condition, is the right one. The second is that we work with the
absolute outcome-summed norm from the start. This converts the signed sum over fault paths, where input-sensitive
contributions could in principle cancel against one another, into a positive polymer gas where no cancellation is
assumed and the bound is therefore robust. A weaker argument that controlled only two-point correlations would miss the high-order terms in which $d$ individually invisible single-wire records multiply into a logical operator. The all-order cluster expansion is designed to control exactly those terms.

The constants are explicit. A \emph{fault-support region} is a connected set of detector cells on the three-dimensional decoding graph, and a region is \emph{charged} when its support carries a homologically nontrivial logical operator and \emph{neutral} otherwise. The neutral quotient is the code-space scalar action of a below-distance operator by Knill--Laflamme. The constant $\ell_\ast$ is the maximum number of data qubits a single detecting cell covers, so a weight-$w$ operator occupies at least $w/\ell_\ast$ cells and a charged region has size at least $d/\ell_\ast$. For the rotated code, $\ell_\ast$ is of order the check weight times the temporal extent. The decay base is $\theta=a=e\nu K_{\rm rod}\etabar$, where $\nu\le e\Delta$ is the connective constant of the bounded-degree detector graph, $K_{\rm rod}$ and $\etabar$ are the persistent-wire constants, and the binding Koteck\'y--Preiss smallness is $ea<1$. All these constants are independent of $d$ and $T$. The spacetime volume enters only as the polynomial prefactor $\mathrm{poly}(d,T)$
from the choice of the region's anchor cell, which is harmless because $T=\Theta(d)$ makes it $\Theta(d^3)$, absorbed by
the exponential $\theta^{d/\ell_\ast}$. The exact transcript factorization enters only by letting PWTS bound each region's activity by a product over the region, $\eta(R)\le(K_{\rm rod}\etabar)^{|R|}$. It is an explicit modelling assumption, not a derived fact, so the device-model derivation of Supplementary Note~2 states it as a hypothesis.

\subsection*{Relation to fault-tolerance threshold theorems}
The structure of the argument mirrors the standard proof of the fault-tolerance threshold, and the parallel identifies the contribution. A threshold theorem assumes a local-stochastic or local-Hamiltonian noise model, expands the
faulty circuit into fault paths, and shows that the probability of an uncorrectable fault cluster is exponentially
small below a constant threshold~\cite{aharonov2008,aliferis2006,terhalburkard2005}. Our argument assumes PWTS, the transcript analogue of that noise condition. It expands the observed instrument into detector-region polymers and shows that the input-sensitive activity of a charged cluster is exponentially small below a constant threshold. The
homological lower bound on charged-cluster size is the same geometric fact that underlies the surface-code
distance~\cite{dennis2002,fowler2012}. The contribution applies standard cluster-expansion machinery~\cite{koteckypreiss1986,fernandezprocacci2007,friedli2017} to a new object. A threshold theorem bounds the \emph{computation}'s failure probability, while this theorem bounds the probability that the \emph{transcript} reveals the input. The two are governed by the same code distance, which is the technical content of the relocation principle, so PWTS is analogous to the assumptions a threshold theorem already makes. It holds for the reset-ancilla transcript model of the same local, leakage-managed hardware required by a threshold theorem. The required conditions are the stated factorization and the absence of an exposed record calibrated to a data-local operator, while PWTS fails for the same kinds of long-range or persistent pathologies. The factorization and the no-data-calibrated-record proviso are the extra transcript assumptions beyond the noise model.

\subsection*{Time scale, composition, and stopping}
\Cref{thm:main} is stated at $T=\Theta(d)$, but real execution logs are far longer than one code distance, so the
permitted growth of $T$ must be made explicit. The bound has the form $C_0\,\mathrm{poly}(d,T)\,\theta^{d/\ell_\ast}$, in which only the prefactor carries $T$, and this dependence on the prefactor produces the three time regimes listed below.

\begin{itemize}[leftmargin=1.4em,itemsep=1pt]
\item For $T=\mathrm{poly}(d)$ the prefactor remains polynomial in $d$ and the bound is still $e^{-\Theta(d)}$.
\item More generally the bound is non-trivial whenever $\mathrm{poly}(d,T)=e^{o(d)}$, i.e.\ for any $T=e^{o(d)}$, so privacy survives execution logs that are subexponentially long in the code distance.
\item At $T=e^{\Theta(d)}$ with a large enough constant the prefactor overwhelms $\theta^{d/\ell_\ast}$ and the statement becomes vacuous, so the theorem makes no claim in that regime, and none should be read into it.
\end{itemize}

For a computation composed of $N$ memory blocks, each protected at $\kappa_d=e^{-\Theta(d)}$, total variation is subadditive over the concatenated transcript, giving $N\kappa_d$, so under this composition bound privacy survives $N=e^{o(d)}$ blocks. This is a union bound and nothing more. It is loose because it ignores that the blocks share a code and a device, and we do not claim tightness for it.

An adversary who chooses \emph{when to stop reading} gains nothing because a stopping rule is a function of the transcript, and the data-processing inequality forbids post-processing from increasing distinguishability. An adversary who influences \emph{what the device does} is a different matter and is excluded by hypothesis. The schedule is fixed and input-independent, so an adaptive or maliciously chosen schedule is outside the scope of this work.

\subsection*{The converse}
The converse is a binary-testing statement for charged calibrated witnesses. Both qualifiers enter its hypothesis. A witness is \emph{charged} when the product of the local Paulis it reports is a nontrivial logical operator rather than a stabilizer. A witness is \emph{calibrated} when each report tracks its Pauli with a known, bounded error, so that the adversary can combine the reports without first learning the device.
The correct primary statement is about the \emph{estimator}, not about the reports it is built from.
\begin{lemma}[charged calibrated witness]\label{lem:converse}
Let $L$ be the distinguishing logical operator. Suppose some function $f(U)$ of the transcript estimates the eigenvalue
of $L$ with error probability $p_W$. Then
\begin{equation}\label{eq:converse}
\TV(\Pzero,\Pone)\ge 1-2p_W .
\end{equation}
\end{lemma}
The proof is Neyman--Pearson applied to the test that outputs $f$, and is given in Supplementary Note~1. Everything
device-specific is then isolated in a single quantity, $p_W$, and the work is to bound it for a named witness.

For a lattice-surgery $\Zbar_1\Zbar_2$ measurement, $f$ is the \emph{decoded} merged-boundary parity, obtained by
fault-tolerant decoding of the redundant spacetime syndrome of the merge. Below threshold this decoder achieves
$p_W=p_{\rm LS}\le Ce^{-\alpha d}$, with the same code-distance protection as a surface-code
memory~\cite{horsman2012,fowler2012,dennis2002}, giving $\TV\ge1-e^{-\Theta(d)}$.

\begin{corollary}[raw local reports; and why this route is weak]\label{cor:rawreports}
If instead the transcript carries $m$ \emph{undecoded} calibrated reports $g_j(U)$ of local Paulis $P_j$ with
$\prod_j P_j=s\,L$ ($s\in\Sgroup$) at readout errors $\delta_j$, then $f=\prod_j g_j$ estimates $L$ with
$p_W\le\sum_{j=1}^m\delta_j+p_{\rm FT}$, and \cref{eq:converse} gives
$\TV\ge1-2\sum_j\delta_j-2p_{\rm FT}$.
\end{corollary}
This corollary is included because it is the naive route and it is instructive that it fails. A charged witness on a
distance-$d$ patch generally needs $m=\Theta(d)$ reports, so at fixed per-report error $\delta$ the union bound gives
$\sum_j\delta_j=\Theta(d\delta)$ and the lower bound \emph{degrades} with distance, becoming vacuous once
$d\delta\gtrsim 1/2$. The bound is therefore useful only for short witnesses or for $\delta=o(1/d)$.

The distinction is not academic, and it is visible in hardware. At fixed merge duration, the measured merged-parity error is $p_{LS}=0.100$ at $d=3$ against $0.217$ at $d=7$ (Hardware measurement section). Because a linear-geometry merge has a single stabilizer seam at every distance and no spatial protection, this error \emph{rises} with distance exactly as \cref{cor:rawreports} would predict for an unprotected multi-report witness. The exponential converse is a statement about a decoded witness on a two-dimensional seam, and it must be stated that way. Reading \cref{cor:rawreports} as if it delivered $1-e^{-\Theta(d)}$ would be an error.

The two qualifiers distinguish the converse from an inversion of the no-leakage theorem. A leaking witness is both charged and calibrated. A bare PWTS violation may instead measure the wrong Pauli or carry a one-time pad that removes correlation with every fixed logical operator.

\subsection*{Stating the prediction per qubit}

The exact transcript laws used above are evaluated at a single $\gamma$, which assumes every data qubit decays at the
same rate. On hardware they do not, and for a leak of order $\gamma^{d_Z}$ that assumption is not a small correction
once $d_Z>1$.

The reason is structural. At small exposure, a $\gamma^{d_Z}$ leak is dominated by the $d_Z$ qubits that decay fastest, so their contribution is weighted far above their share of the line. At large exposure, those same qubits have already decayed and stop contributing. A law evaluated at the line's median $\gamma$ therefore \emph{understates} the prediction at the short end and overstates it at the long end. A measurement compared against this law appears to exceed the theory at small exposure and fall short of it at large exposure. That crossing is a signature of the
evaluation, not of the device.

We encountered exactly this. A $[[4,1,2]]$ exposure sweep whose four data qubits were recorded at $T_1=234$--$\SI{250}{\micro\second}$ at submission sat a factor $1.52$ above its median-$\gamma$ law at the shortest exposure. This discrepancy was large enough that the $95\%$ one-sided lower bound exceeded the law at two of five exposures, which a lower bound cannot do if damping is the whole leak. One calibration later the same four qubits read $180$, $278$, $100$ and $\SI{260}{\micro\second}$, and re-evaluating the same exact law with a per-qubit rate vector, $\gamma_i(\tau)=1-\exp[-(t_{\rm round}+\tau)/T_1^{(i)}]$, brings the shortest-exposure ratio to $0.97$ and removes both lower-bound violations. No mechanism beyond amplitude damping is required to account for this discrepancy.

Two practices follow, and both are used for the measurements reported here, the first being that the prediction is evaluated on the rate vector instead of a summary statistic of it. The enumeration cost is unchanged because the law is already computed by recursion over rounds and stabilizer branches. Second, the per-qubit $T_1$, $T_2$ and readout error are snapshotted both at submission and after the job returns because they move. Across a single job, the worst-drifting qubit changed by more than a factor of two, so the two snapshots bracket what actually applied, and where they disagree the $\gamma$ axis is reported as a range instead of a number.

\subsection*{Noise models and simulators}
All numerics use the rotated surface code with circuit-level depolarizing noise (one- and two-qubit depolarization, measurement and reset flips, and per-round data depolarization)~\cite{tomita2014} at a stated physical rate, decoded where needed by minimum-weight perfect matching~\cite{edmonds1965,higgott2022,delfosse2021,bravyi2014}. Leakage follows standard leakage-reduction analyses~\cite{aliferis2007,suchara2015,brown2020}. Every result carries a bootstrap
confidence interval, every figure is regenerated by a single script from a committed result file, and random seeds and
tool versions are pinned in each file. Four instruments are used, and since the value of an adversarial test depends on
its faithfulness we state the construction of each together with what it can and cannot establish.

\emph{The adversarial classifier.} Features are the \emph{entire} detector-event vector of the rotated-memory circuit, not a chosen summary, and a defect count is one example of such a summary, so the feature set discards nothing. The resolution of the null is then set by the score class and not by the features, and a regularized logistic score reads the linear structure of that vector. The two classes are the same circuit with and without a logical-$\Xbar$ injection verified to be detector-silent and of odd overlap with $\Zbar$. That construction rules out a mislabelled pair, so a chance result here reflects the algebraic identity of the two laws and not a failure to label them. A regularized logistic model is trained using half of the shots, and we report the held-out area under the ROC curve with the bootstrap interval over the test set. This is a finite-sample monotone proxy for the distinguishing advantage the no-leakage theorem bounds, with null value one half. The biconditional with $\TV=0$ holds only for the optimal likelihood-ratio score, and the realisable AUC excess is bounded in Supplementary Note~20.6.

\emph{The hard-regime simulator.} The hard-regime simulation uses an exact statevector trajectory at $d=3$, not a Pauli-twirled approximation, while the $d=3$ code is held as a state vector and the logical input is prepared by projection onto the code space. Rounds of projective stabilizer measurement alternate with a noise step that applies a genuine coherent rotation $\exp(-i\theta/2(\cos\varphi X+\sin\varphi Z))$ to each data qubit. With fixed probability, the same noise step promotes a qubit to a persistent leaked carrier that applies a random Pauli each round until a leakage-reduction reset returns it. These are
real non-Clifford rotations, whose effect a stabilizer sampler cannot capture.

\emph{The converse simulator.} It decodes the distance-$d$-protected merge seam with minimum-weight perfect matching and reports the logical failure rate, which is the merged-parity error. This is exactly the rate that the converse turns into a leak, and it comes from a standard named decoder instead of an idealized bound.

\emph{The PWTS calibration}, and it is the weak one. It builds the detector graph from the error model, finds the
connected components of fired detectors in each shot, and fits the geometric tail of their size distribution. That is a measurable diagnostic for the \emph{geometric form} PWTS assumes, but it is only a non-rigorous proxy for the pre-logical polymer activity $\eta(R)$. This activity is a resummed norm over fault-path contributions, not a fired-detector frequency, and the diagnostic does not provide a certified bound on it. Confirming a geometric tail below one does not establish the stronger
Koteck\'y--Preiss smallness $ea<1$ that the closed bound requires. At the simulated rates that smallness is not met, so
the accessible-rate evidence for privacy is the direct distinguishing-advantage test, with the closed bound holding in
the deeper sub-threshold regime (Supplementary Note~44.2).

\FloatBarrier
\section*{Data availability}
No external or third-party datasets were used. Two kinds of data underlie the figures, and both are in the public
archive at \url{https://github.com/Mercury0828/qec-transcript-privacy-artifact}, with release v1.0.1 archived
at \archivedoi{} so that the citation resolves to that fixed snapshot rather than to a moving branch. The simulation
results are one JSON per
figure panel, recording per-run outcomes, random seeds and tool versions; the analysis plan, circuit families and
result manifests for the original grouped sweep were fixed before any quantum-processor time was used. The
randomised re-acquisition and the fixed-code exposure sweeps are follow-up experiments, motivated by the
acquisition-order analysis of the first sweep and by the absence of an exponent measurement in it; their circuit
families, shot allocations and analysis procedures were frozen before the corresponding follow-up jobs were
submitted, and the frozen plans are in the repository beside the results. We describe the two stages separately
rather than presenting the whole programme as pre-specified. The hardware results are the raw per-shot registers from
\texttt{ibm\_cleveland} (one compressed \texttt{.npz} per circuit) together with a metadata record per session giving
the job identifiers, the physical qubit lines, the transpiled depths, the inter-shot delay, and the device calibration
snapshot read at submission time.

\section*{Code availability}
The code that generates and analyses those results, namely the Stim circuit construction, minimum-weight perfect-matching
decoding, the adversarial-classifier and exact statevector trajectory simulators, the cluster-tail calibration, and the
figure scripts, is available in the same archive at
\url{https://github.com/Mercury0828/qec-transcript-privacy-artifact}.

\begin{sloppypar}
Five scripts exist specifically so that a reader need not take the constants on trust.
\texttt{sim/hw/coset\_distances.py} computes $d_X$, $d_Y$ and $d_Z$ for every code family used here by exact
enumeration of the three logical cosets, with no heuristic search; \texttt{sim/hw/damping\_order.py} builds the exact
transcript law under amplitude damping by full density-matrix evolution with every outcome branch enumerated, and
measures the order in $\gamma$ without sampling, with and without the environment jump record;
\texttt{sim/hw/oaqec\_check.py} evaluates the operator-algebra correctability condition for the logical
$\Zbar$-algebra; \texttt{sim/hw/threeway\_tv.py} recomputes the hardware bound under a train/validate/test split
with exact one-sided Clopper--Pearson confidence bounds; and \texttt{sim/hw/audit\_si\_constants.py} re-derives
every Koteck\'y--Preiss side condition, threshold, tail sum, measured slope, certified bound and certificate
number quoted in the text, and fails if any of them does not reproduce. It reports $170/170$ checks passing. Supplementary Note~11.12 states which links in the argument these cover,
which are cited results used as stated, and which remain ordinary mathematics for a reader to check.
\end{sloppypar}

\section*{Author contributions}
J.S. and H.Z. conceived the study and developed the theoretical framework. J.S. carried out the proofs, implemented the
simulations, and drafted the manuscript. H.Z. contributed to the privacy formulation, supervised the project, obtained
and administered the quantum-processor allocation, and revised the manuscript. Both authors discussed the results and
approved the final version.

\section*{Acknowledgments}
% If the Miami University Quantum Program Coordinator (CECHelp@MiamiOH.edu) supplies required acknowledgement
% wording for the Cleveland Clinic allocation, substitute it verbatim for the first sentence below.
% See sim/hw/CCQC_SETUP.md.
We acknowledge the use of IBM Quantum services for this work. Quantum-processor time on \texttt{ibm\_cleveland} was
provided through the Cleveland Clinic IBM Quantum System One allocation administered by Miami University and accessed
under the affiliation of H.Z. The views expressed are those of the authors and do not reflect the official policy or
position of IBM or the IBM Quantum team. The authors received no specific grant from funding agencies in the public,
commercial, or not-for-profit sectors.

\section*{Competing interests}
The authors declare no competing interests.

\bibliographystyle{unsrtnat}
\bibliography{refs}

\end{document}